\documentclass[11pt,cls,onecolumn]{IEEEtran}

\usepackage{stfloats}
\usepackage{color}
\usepackage{amssymb} 
\usepackage{mathtools}
\usepackage{footnote}
\usepackage{amsfonts}
\usepackage[colorlinks=true, linkcolor=blue]{hyperref}
\usepackage{booktabs}
\usepackage{makecell}
\usepackage{graphicx}
\usepackage{subcaption}
\usepackage{stmaryrd}  

\newcommand{\ve}[1]{{\bf #1}}

\newcommand{\mc}{\mathcal}
\newcommand{\wcf}{\mathcal F_m^{\text{W}}}
\newcommand{\bfm}{\mathbb {BF}_m}
\newcommand{\intset}[1]{[\![#1]\!]}
\newcommand{\lt}{\mathsf{T}}
\newcommand{\lf}{\mathsf{F}}
\newcommand{\cbfc}{C_{\text{BFC}}}

\newcommand{\rf}{rate function }

\newtheorem{lemma}{Lemma}
\newtheorem{theorem}{Theorem}
\newtheorem{proposition}{Proposition}
\newtheorem{definition}{Definition}

\newtheorem{remark}{Remark}

\newtheorem{example}{Example}

\begin{document}

\sloppy
\IEEEoverridecommandlockouts
\title{Capacity regimes for Boolean function computation via channels}

\author{ Jingge Zhu, Matthias Frey\\


%
\thanks{Parts of this work have been presented at ISIT 2026 \cite{BFC_ISIT, BFC_ISIT_arXiv}. J. Zhu and M. Frey are  with Department of Electronic and Electrical Engineering, The University of Melbourne,  Parkville, Victoria, Australia. (e-mail: \{jingge.zhu, matthias.frey\}@unimelb.edu.au)} 
%
%
}




\maketitle

\begin{abstract}
Consider a point-to-point communication system in which the transmitter holds a binary message of length $m$ and transmits a corresponding codeword of length $n$.  The receiver's goal is to recover a Boolean function of that message, where the function is unknown to the transmitter, but chosen from a known class $\mc F$. We are interested in the asymptotic relationship of $m$ and $n$:  given $n$, how large can $m$ be (asymptotically), such that the value of the Boolean function can be recovered reliably? This problem generalizes the identification-via-channels framework introduced by Ahlswede and Dueck. In this paper, we formulate the notion of computation capacity, and derive achievability and converse results for a large class of functions $\mc F$, characterized by the Hamming weight of functions. Different from the classical transmission problem, the performance of the function computation problem is jointly characterized by the computation capacity and the rate function, namely how $m$ scales with $n$ asymptotically. Our results  give a complete characterization of the \rf of the computation problem, and provide  upper and lower bounds on the computation capacity, where they differ by a factor of at most $2$.
\end{abstract}

\section{Introduction}

We consider the problem of computing Boolean functions over noisy channels. To motivate the problem, consider the following simple example of a Battery Management System in an electric vehicle,  where a sensing unit reports a binary status vector  $\ve b=(\ve b_1,\ve b_2,\ve b_3)\in\{0,1\}^3$,
with $\ve b_1$, $\ve b_2$, and $\ve b_3$ indicating over--voltage, under--voltage, and
over--temperature conditions of the battery, respectively.  Depending on the operational objective, the controller does not require the full vector $\ve b$, but instead only a Boolean function of $\ve b$, transmitted through a possibly noisy channel.  For example, in \textit{driving mode} the controller needs to evaluate 
\begin{align*}
    f_{\mathrm{drive}}(\ve b):=\ve b_1\lor \ve b_3,
\end{align*}
where $0$ means ``false'', $1$ means ``true'', and $\lor$ is the logical OR operator. The function  $f_{\mathrm{drive}}$ equals $1$ (signals an alarm) whenever either an over--voltage ($\ve b_1=1$) or an over--temperature ($\ve b_3=1$) event occurs. 
In \textit{charging mode}, the controller must determine whether charging is not permissible, 
represented by
\begin{align*}
    f_{\mathrm{charge}}(\ve b): =\ve b_1\lor (\ve b_2\land \ve b_3),
\end{align*}
where $\land$ is the logical AND operator. The function $f_{\mathrm{charge}}$ equals $1$ (signals an alarm)  when either over--voltage $(\ve b_1=1)$, or the abnormal 
``under--voltage and hot'' ($\ve b_2=1$ and $\ve b_3=1$) condition occurs.

The Boolean Function Computation (BFC) problem is an abstraction of this example, where the receiver aims to recover one of several distinct Boolean functions of the same underlying binary sequence. To be precise, consider a point-to-point communication system where the transmitter chooses a message  represented by binary sequences of length $m$, and the decoder chooses a Boolean function $f: \{0,1\}^m\rightarrow \{0,1\}$ from a family $\mc F$. Here we assume that the set $\mc F$ is known to both the encoder and the decoder, but the function $f$ (chosen by the decoder) is unknown to the encoder. Let $i\in\{0,1\}^m$ denote the message at the transmitter, which is encoded to a codeword  $X^n $ of length $n$ and is transmitted through a noisy channel. With the channel output $Y^n$, the decoder is required to reliably decide if $f(i)=0$ or $f(i)=1$.

The BFC problem is a generalization of the \textit{identification via channel} problem studied in a seminal paper by Ahlswede and Dueck~\cite{ahlswede_identification_1989}.  The goal of identification is for the decoder to  decide if the transmitted message is  \textit{a particular one}. It is easy to see that the identification problem is equivalent to an instance of BFC problem with a suitable choice of function class $\mc F$ (details in Section \ref{sec:example_BFC}). A major result from \cite{ahlswede_identification_1989} is that for the identification problem, the length of the message $m$ can scale \textit{exponentially} with $n$, namely $m\sim 2^{n}$, as compared to the classical Shannon transmission problem where $m$ can only scale linearly with $n$, $m\sim n$. After the appearance of~\cite{ahlswede_identification_1989} (and \cite{identification_feedback_1989} where feedback is considered),  many works have extended the identification problem to other set-ups. For point-to-point channels, a line of more recent work including \cite{salariseddigh_deterministic_2022} \cite{colomer_deterministic_2025}, investigates the identification problem when the encoders are restricted to a deterministic function of the message, for which the original work~\cite{ahlswede_identification_1989} only gives a very brief treatment. Explicit code constructions for identification are investigated in   \cite{verdu_explicit_1993} \cite{ahlswede_identification_1991} and further improved in  \cite{kurosawa_strongly_1999} and  \cite{gunlu_code_2022}. Extending the results to multi-user scenarios is considered in \cite{ahlswede_general_2008} for multiple-access channels and \cite{bracher_identification_2017} for broadcast channels. The readers are referred to the recent survey \cite{von_lengerke_codes_2025} for a more comprehensive summary of known results. 

Particularly relevant to our work is \cite{ahlswede_general_2008}, where a general information transfer problem is formulated. While a specialization of this model is equivalent to the Boolean function computation problem, only partial results are obtained in \cite{ahlswede_general_2008}. We make the connection and comparison explicit in Section \ref{sec:Ahlswede_models}.

The main contribution of this paper is to characterize the possible scaling law of $m$ and $n$ for the BFC problem, and provide tight upper and lower bounds on the computation capacity. We show that the scaling  behaviour heavily depends on the property of the function class $\mc F$. Specifically, we consider $\mc F$ consisting of functions $f$ that have a \textit{Hamming weight} equal to, or,  smaller than or equal to a certain threshold, where the Hamming weight is defined to be the cardinality of the pre-image of $1$ under $f$, an integer between $0$ and $2^m$. We denote the hamming weight by $S(m)$, as it is generally a function of $m$.  Our results naturally divide $S(m)$ into three regimes. We call $S(m)$ \textit{small} if $\log S(m)=O(\log m)$, \textit{medium} if $\log S(m)=\omega(\log m)$ but $\log S(m)\leq \alpha m$ for some $\alpha\in(0,1)$, \textit{large} if $\log S(m)\rightarrow m$. Our results show that $m$ scales exponentially with $n$ if $S(m)$ is small (as in the identification problem, which corresponds to $S=1$), while it scales linear with $n$ when $S(m)$ is large (the same behavior as in Shannon's transmission problem). In the medium regime, $m$ can exhibit various scaling behaviours with respect to $n$, including linear, quasi-linear, polynoimal, or sub-exponential, as shown by  examples in  Table \ref{table:summary}. The exact form the the \textit{rate function} is characterized in Theorem~\ref{thm:capacity_bfc}, which also provide upper and lower  bounds on the computation capacity, where the upper and lower bounds differ by a constant smaller than or equal to $2$.  Furthermore, the computation capacity is shown to be exactly $C$ in some special cases, where $C$ denotes the Shannon capacity of the channel.

Given a positive integer $M$, we use $\intset{M}$ to denote the set of integers $\{1,\ldots, M\}$. For a given set $\mc X$, we use $|\mc X|$ to denote the cardinality of the set.

\section{Problem formulation, auxiliary results, and examples}
We consider a channel with the input alphabet $\mc X$ and the output alphabet $\mc Y$.  A channel $W(\cdot|x^n)$ is viewed as a conditional distribution (transition kernel) over $\mathcal Y^n$, given the input codeword $x^n\in\mathcal X^n$. We use $\bfm$ to denote the set of Boolean functions of $m$ inputs, $\bfm :=\{f: \{0,1\}^m\rightarrow \{0,1\}\}$. Consider the communication problem where a \textit{binary} message, represented by a binary sequence of length $m$, is chosen and transmitted after encoding by the transmitter. A Boolean function $f:\{0,1\}^m\rightarrow \{0,1\}$, unknown to the transmitter, is chosen by the receiver, with the goal to reliably recover the function value of the transmitted message. The notion of Boolean function computation (BFC) code is formally defined as follows.

\begin{definition}[Boolean function computation (BFC) code]
Let  $m, n$ be positive integers. An $(n,m,\mc F, \lambda_1, \lambda_2)$ Boolean function computation (BFC) code consists of encoders $(Q_i)_{i \in \{0,1\}^m}$ and decoding sets $(D_j)_{j \in \intset{|\mc F|}}$, such that
\begin{itemize}
\item $\mathcal F\subseteq\bfm$ is a set of Boolean functions with $m$ inputs
\item $Q_i$ is a probability distribution on $\mathcal X^n$ for all $i\in\{0,1\}^m$
\item $D_j\subset \mathcal Y^n$ for all $j\in \intset{|\mc F|}$
\item  false negative error: $\ve 1_{f_j(i)=1}Q_iW(D_j^c)\leq \lambda_1$ for all $i, j$
\item  false positive error: $\ve 1_{f_j(i)=0}Q_iW(D_j)\leq \lambda_2$ for all $i, j$
\end{itemize}
where we define  $QW(D):=\int W(D|x)Q(dx)$.
\end{definition}

Let $f^{-1}[1]$ denote the preimage of $1$ under $f$. The cardinality of the preimage $f^{-1}[1]$ is called the Hamming weight of the Boolean function $f$. We consider the class of  \textit{constant weight Boolean functions} defined as
\begin{align*}
\wcf(S) := \{f\in \bfm: |f^{-1}[1]|= S\}
\end{align*}
where  $S$ can take integer values between $0$ and $2^m$, and 
$|\wcf(S)|= {2^m\choose S}$.
Similarly, we define the set of functions whose Hamming weight is \textit{smaller or equal to} $S$ as
\begin{align*}
\wcf(\leq S):= \{f\in \bfm: |f^{-1}[1]|\leq  S\}
\end{align*}

As we will see in the sequel, depending on the relationship between  $S$  and $m$, the maximal message length that can be supported exhibits different scaling behaviour with respect to $n$. To this end, we define the achievable computation rate as follows.

 \begin{definition}[Achievable computation rate]\label{def:rate}
Let $m,n$ be  positive integers,  $L:\mathbb R\times \mathbb N\rightarrow \mathbb R^+$ be the \textit{rate function}, and consider a sequence of sets of Boolean functions $\{\mc F_m\}_m:=\{\mc F_m, m=1,2,\ldots\}$. We say $R$ is \textit{an achievable computation rate} for the sequence $\{\mc F_m\}_m $  with the \rf $L$, if for all $\lambda_1,\lambda_2>0$, all  $\eta>0$, there exists $n_0$ such that for all $n\geq n_0$, there exists an integer $m$ and an  $(n,m, \mc F_m,  \lambda_1,\lambda_2)$ BFC code satisfying
\begin{align*}
m \geq  L(R-\eta, n).
\end{align*}
\end{definition}

To relate this to more familiar notions, the choice $L(R, n) = Rn$ corresponds to the case when $m$ scales linearly with $n$ (as in the classical transmission problem). The choice $L(R,n)=2^{Rn}$ corresponds to the case when $m$ scales exponentially with $n$ (as in the identification problem). We will see other rate functions in our main results in the next section.

\begin{definition}[Computation capacity]
For a given \rf  $L$ and a sequence of sets of Boolean functions $\{\mc F_m\}_m$, the supremum of the achievable computation rate for  the sequence $\{\mc F_m\}_m$ is called the (Boolean function) computation capacity of the channel for $\{\mc F_m\}_m$ with the \rf $L$.
\end{definition}

\subsection{Symmetry of $\wcf(S)$}

In this paper, we will study the computation capacity for both classes $\wcf(S)$ and $\wcf(\leq S)$. We first make the observation that  the capacity behaviour for $\wcf(S)$ and $\wcf(\leq S)$ is different for large $S$. The reason is that $|\wcf(\leq S)|$ is an increasing function of $S$, whereas $|\wcf(S)|$ is symmetric around the mid-point $S=\frac{1}{2}2^m$.  In particular, the set $\wcf(S)$, whose elements satisfy $|f^{-1}[1]|=S$ and $|f^{-1}[0]|=2^m-S$, essentially has the same property as the set $\wcf(2^m-S)$ by flipping the output of the functions in the former set. We formalize the observation in the following Lemma.
\begin{lemma}[Symmetry of $\wcf(S)$]
Given an $(n,m, \wcf(S), \lambda_1,\lambda_2)$ BFC code, we  can construct an $(n,m, \wcf(2^m-S), \lambda_2,\lambda_1)$ BFC code.
\label{lemma:preimage1}
\end{lemma}
\begin{proof}
Given the encoder $\{Q_i\}$ and decoder $\{D_j\}$ of the $(n,m, \wcf(S), \lambda_1,\lambda_2)$ BFC code, and for each function $f_j\in\wcf(S)$, we define a new set of functions $\tilde f_j=1-f_j$,  and the corresponding encoders and decoders as
\begin{align*}
\tilde Q_i&:=Q_i\\
\tilde D_j&:=D_j^c
\end{align*}
Clearly, the set $\{\tilde f_j\}_j=\wcf(2^m-S)$, and it  holds $\tilde f_j^{-1}[1]= f_j^{-1}[0]$, $\tilde f_j^{-1}[0]= f_j^{-1}[1]$. To characterize the error probability of the new code, we have
\begin{align*}
\ve 1_{\tilde f_j(i)=1}\tilde Q_iW(\tilde D_j^c)=\ve 1_{f_j(i)=0}Q_iW(D_j)\leq \lambda_2\\
\ve 1_{\tilde f_j(i)=0}\tilde Q_iW(\tilde D_j)=\ve 1_{f_j(i)=1}Q_iW(D_j^c)\leq \lambda_1
\end{align*}
where the inequalities follow from the definition of the $(n,m, \wcf(S), \lambda_1,\lambda_2)$ BFC code.
\end{proof}

\begin{remark}
This lemma also implies that when we consider  the computation capacity for the class $\wcf(S)$, we can without loss of generality restrict $S$ to the range $[0,2^{m-1}]$, as any $S>2^{m-1}$ will have the same capacity (and rate function) as $2^m-S$, which is smaller than or equal to $2^{m-1}$.
\label{remark:symmetry}
\end{remark}


\subsection{Equivalent rate functions}

Since the achievable rate is an asymptotic notion, the same computation rate can be described by different rate functions. For example, the rate functions $L(R,n)=Rn$ and $L'(R,n)=Rn+a\geq 0$ for some constant $a$  is asymptotically equivalent, in the sense that if $R$ is achievable with $L$, it is also achievable with  $L'$. On the other hand, we may prefer  $L$ over $L'$ for simplicity. In light of this observation, we formulate the notion of equivalent rate functions in this section.   We first define the notion of valid rate functions, a very mild condition which all ``reasonable" rate functions  satisfy.

\begin{definition}[Valid rate function]
A function $L:\mathbb{R}_+\times\mathbb{N}\rightarrow\mathbb{R}_{+}$ is called a \textit{valid rate function} if for every $R_1>R_2>0$,
\begin{align}
    \liminf_{n\rightarrow\infty}
    \frac{L(R_1,n)}{L(R_2,n)}
    >1.
    \label{eq:valid-rate-function}
\end{align}
\label{def:valid_rate_function}
\end{definition}
Roughly speaking, this condition requires that any fixed increase in the rate parameter eventually produces a non-vanishing multiplicative increase in the rate function, at least for large enough $n$. Now we can define the notion of equivalent rate functions.

\begin{definition}[Equivalent rate functions]
Two valid rate functions $L_1$ and $L_2$ are called equivalent if for all $R>0$, 
\begin{align}
    \lim_{n\rightarrow\infty}
    \frac{L_1(R,n)}{L_2(R,n)}
    =1.
    \label{eq:equivalent-rate-functions}
\end{align}
\label{def:equivalent_rate_function}
\end{definition}

\begin{lemma}[Invariance under equivalent rate functions]
Let a sequence of sets of Boolean functions $\{\mc F\}_m$ be given, and let $\cbfc$ denote the computation capacity.  Let $L_1$ and $L_2$ be two equivalent valid rate functions. We have the following statement:
\begin{itemize}
\item If $R$ is an achievable computation rate for $\{\mc F\}_m$ with the \rf $L_1(R,n)$, then $R$ is also an achievable computation rate for $\mc F$ with the \rf $L_2(R,n)$.
\item If $\cbfc\leq A$ with the \rf $L_1(R,n)$, then we also have $\cbfc\leq A$ with the \rf $L_2(R,n)$.
\end{itemize}
\label{lemma:equivalence}
\end{lemma}

The proof of this result in given in Appendix \ref{appendix:proof_equivalence}. This result shows that if two valid rate functions are equivalent, we can state the same capacity result with either of them. For each example, it is easy to check that the two rate functions $L(R,n):=Rn$ and $L'(R,n):=Rn+a$ are valid and equivalent rate functions. In general, the usefulness of the result will be demonstrated in Proposition \ref{prop:special_S}, in particular in Case 6) where a direct calculation of rate function given in Theorem \ref{thm:capacity_bfc} is difficult.

\subsection{Examples: $\wcf(S)$ and $\wcf(\leq S)$  for different $S$}
\label{sec:example_BFC}

Before presenting the main results, we first give examples of the function classes $\wcf(S)$ and $\wcf(\leq S)$ for different choices of $S$, and comment on connections to semantic evaluations in propositional logic, and other existing problem formulations in the literature.

In this subsection, we use the notation $\ve b =(\ve b_1,\ve b_2,\ldots, \ve b_m)$ where each $\ve b_i$ is a Boolean variable.

\begin{example}[$\wcf(c)$]\label{example:S_constant} Let $c$ be a constant not depending  on $m$. A function in the class $\wcf(c)$ evaluates to $1$ on exactly $c$ input sequences. A special case is  $c=1$, where the functions in the class  $\wcf(1)$ can be expressed explicitly as
\begin{align}
f_j^{id}(\ve b) := \prod_{i\in A_j}\ve b_i \cdot \prod_{i\in A_j^c}(1-\ve b_i)
\label{eq:f_id}
\end{align}
for all subsets $A_j\subseteq \intset{m}, j=1,\ldots, 2^m$. Computing $\wcf(1)$ is equivalent to  the  \textit{identification via channel} problem \cite{ahlswede_identification_1989} (more details in Section \ref{sec:Ahlswede_models}). Indeed, $f_j^{id}(\ve b)=1$ if and only if $\ve b$ is the unique message that corresponds to the set $A_j$ (and there are in total $2^m$ of them).
\end{example}

\begin{example}[$\wcf(\leq cm^{\beta})$] Let $c>0$ and $\beta$ a positive integer not depending on $m$. If we ask the question: ``Are there exactly $\beta$ inputs equal to $1$"? This question can be formulated by the Boolean function
\begin{align*}
h_\beta(\ve b) := \ve 1_{\sum_{i=0}^m \ve b_i=\beta}
\end{align*}
The weight of this function is ${m\choose \beta}$ which is upper and lower bounded as $\frac{m^\beta}{\beta^\beta}\leq {m\choose \beta}\leq \frac{m^\beta}{\beta!}$. Hence this function belongs to the class $\wcf(\leq cm^{\beta})$ for the choice $c=1/\beta!$.  Similarly, the question ``are there at most $\beta$ inputs equal to $1$" corresponds to the Boolean function
\begin{align*}
\tilde h_{\beta}(\ve b):=\ve 1_{\sum_{i=0}^m \ve b_i\leq \beta}
\end{align*}
These functions are known as the \textit{threshold function}, whose weight is upper bounded as $\sum_{i=1}^\beta {m\choose i}\leq (\frac{e}{\beta})^\beta m^\beta$. They belong to the set $\wcf(\leq cm^{\beta})$ with $c=(\frac{e}{\beta})^\beta$.
\end{example}

\begin{example}[$\wcf(c2^{\gamma m})$]\label{example:t_bit}  Let $c$ and $\gamma\in(0,1]$ be some constant not depending on $m$. If we want to recover the $t$-th bit of the message, we can use the function
\begin{align}
f_t^{bit}(\ve b) := \ve b_t
\label{eq:f_tbit}
\end{align}
It is easy to see that the set of functions $f_t^{bit}, t=1,\ldots, m$ belongs to the  class $\wcf(\frac{1}{2}2^{m})$. As another example, consider the function 
\begin{align*}
f_{S_k}^{\text{AND}}(\ve b):= \prod_{i\in S_k}\ve b_i
\end{align*}
where $S_k\subseteq\{1,\ldots, m\}$ with  cardinality $k$. This function performs the logic AND operation on a subset (of size $k$) of inputs. The Hamming weight of this class of function is $2^{m-k}$. When $k \geq (1-\gamma)m$, these functions belong to the class $\wcf(\leq 2^{\gamma m})$.
\label{example:f_and}
\end{example}

\begin{example}[ranking]\label{example:ranking}
Let $int:\{0,1\}^m\rightarrow \mathbb N$ be the function that maps a binary sequence to an integer, where the sequence represents the binary expansion of that integer. In other words
\begin{align}
int(\ve b):= \sum_{i=1}^m b_i2^{m-i}
\label{eq:int}
\end{align}
Then we can define the ranking function
\begin{align}
f_{r}^{\text{rank}}(\ve b) := \ve 1_{int(\ve b)\leq r}
\label{eq:f_order}
\end{align}
for $r=0,\ldots, 2^m-1$. It is clear that the set of functions $\{f_r^{\text{rank}}, r=0,\ldots, S-1\}$ belongs to the set $\wcf(\leq S)$, for any $S=1,\ldots, 2^m$.
\end{example}

\subsection{Formulae in propositional logic}
The notion \textit{semantic communication} roughly refers to the communication problems where only the ``meaning", or semantics of messages is to be communicated, rather than the entirety of the message. Semantics exist in the context of languages,  with  \textit{propositional logic} being perhaps the simplest form, where there are only two semantic values, $\mathsf{True}$ (denoted by $\lt$) and $\mathsf{False}$ (denoted by $\lf$). In this subsection we comment on the connection between Boolean function computation and the semantic evaluation in propositional logic.

In the setup of propositional logic \cite{enderton_mathematical_1972},  $\mc P=\{p_1,\dots,p_m\}$ denotes a finite set of \emph{propositional atoms}. The set of \emph{well-formed formulae} \(\mc L(\mathcal{P})\) is defined recursively as the smallest set of finitely long tuples of elements of $\mc P \cup \{\neg, \lor, \land, (, )\}$ with the properties:
\begin{itemize}
    \item for every $p \in \mc P$, we have $p \in \mc L(\mathcal{P})$,
    \item if $\varphi \in \mc L(\mathcal{P})$, then $\neg (\varphi) \in \mc L(\mathcal{P})$,
    \item if $\varphi \in \mc L(\mathcal{P})$ and $\psi \in \mc L(\mathcal{P})$, then $(\varphi) \lor (\psi) \in \mc L(\mathcal{P})$ and $(\varphi) \land (\psi) \in \mc L(\mathcal{P})$,
\end{itemize}
where \(p_i\in\mathcal{P}\), and \(\neg\), \(\land\), \(\lor\) are the logical connectives \emph{negation}, \emph{conjunction}, and \emph{disjunction}. For convenience, we do not always write out all the parentheses of every formula (when the order of operations is clear by common conventions), and we also define the connectives
\begin{itemize}
    \item $\varphi \to \psi := \neg \varphi \lor \psi$ \emph{(implication)}
    \item $\varphi \leftrightarrow \psi := (\varphi \land \psi) \lor (\neg \varphi \land \neg \psi)$ \emph{(equivalence)}
    \item $\varphi \oplus \psi := (\varphi \land \neg \psi) \lor (\neg \varphi \land \psi)$ \emph{(exclusive or)}.
\end{itemize}

A \emph{truth assignment} is a function \(\tau:\mc P\to\{\lt,\lf\}\) assigning a truth value to each propositional atom.
The \emph{semantics} of formulae is given by the evaluation map \(v_\tau : \mc L(\mathcal{P})\to\{\lt, \lf\}\), defined recursively via
\begin{align*}
v_\tau(p_i) &= \tau(p_i)\\
v_\tau(\neg \varphi)&=\begin{cases}
\lt \quad \text{if } v_\tau(\varphi) = \lt\\
\lf\quad \text{otherwise}
\end{cases} \\
v_\tau(\varphi\land\psi) &=\begin{cases}
\lt\quad \text{if } v_\tau(\varphi) = \lt \text{ and }  v_\tau(\psi)= \lt\\
\lf \quad \text{otherwise}
\end{cases}\\
v_\tau(\varphi\lor\psi) &=\begin{cases}
\lt\quad \text{if } v_\tau(\varphi) = \lt \text{ or }  v_\tau(\psi)= \lt\\
\lf \quad \text{otherwise.}
\end{cases}
\end{align*}
In this way, each formula \(\varphi\) defines a Boolean-valued function of its atomic propositions under the truth value assignment $\tau$.   We call two formulae $\varphi_1$ and $\varphi_2$ \textit{tautologically equivalent}, if they have the same truth value on all truth assignments of the atoms, namely $\forall \tau, v_{\tau}(\varphi_1)=v_{\tau}(\varphi_2)$. Otherwise we call them \textit{tautologically distinct}. With $m$ atoms, there are infinitely many formulae, but there are exactly $2^{2^m}$ tautologically distinct formulae, with $2^m$ being the number of possible truth value assignments.

The key observation is that Boolean function computation can be viewed as evaluating the semantics of a chosen formula by the receiver, where the truth assignment is determined by the transmitter. Specifically, let $\ve b = (\ve b_1,\ldots \ve b_m)\in\{0,1\}^m$ denote the message at the transmitter, and we associate each coordinate \(b_i\) with an atom \(p_i\in\mathcal{P}\) and define the truth assignment \(\tau_{\ve b}(p_i)=\lt\) if \(\ve b_i=1\) and \(\tau_{\ve b}(p_i)=\lf\) if \(\ve b_i=0\). Any Boolean function \(f\in\bfm\) can then be represented by a propositional formula \(\varphi_f\) such that for every input \(\ve b\in\{0,1\}^m\),
\begin{align*}
f(\ve b)=1 \;\;\Longleftrightarrow\;\; v_{\tau_{\ve b}}(\varphi_f)=\lt
\end{align*}
Thus, evaluating the Boolean function \(f\) on the message $\ve b$ is equivalent to evaluating the propositional formula \(\varphi_f\) under the truth assignment \(\tau_{\ve b}\). Furthermore, there is a one-to-one mapping between the $2^{2^m}$ different Boolean functions and $2^{2^m}$ tautologically distinct formulae. The logic \text{AND} function defined in Example \ref{example:f_and} gave an example of this connection. As another simple example, the formula  that calculates the parity of three atoms
\begin{align*}
\varphi_{\text{parity}} := p_1 \oplus p_2 \oplus p_3
\end{align*}
corresponds to the Boolean function
\begin{align*}
f_{\varphi_{\text{parity}} }(\ve b) := (\ve b_1+\ve b_2+\ve b_3)\pmod 2.
\end{align*}

\begin{example}[Hamming weight of DNF]
Disjunctive normal forms (DNF) are a canonical form of formulae consisting of a disjunction of one or more conjunctions of one or more literals (an atomic or its negation). For example, $(p_1\land p_2\land \neg p_3)\lor p_4$ is in DNF where as $\neg(p_1\land p_2)$ is not. It is known that \cite[Coro. 15C]{enderton_mathematical_1972} in propositional logic, any well-formed formula in $\mc L(\mc P)$ is tautologically equivalent to a formula in DNF. If we have a DNF with $t$ conjunction terms which are mutually exclusive (i.e. no truth assignment can satisfy more than one conjunction term), and each conjunction contains $k$ atoms, then the Hamming weight of this DNF is $t2^{m-k}$. In general if there are overlapping atoms in different conjunction terms, the Hamming weight of this DNF is smaller than or equal to $t2^{m-k_{\text{min}}}$ where $k_{\text{min}}$ is the number of atoms in the shortest conjunction term.
\end{example}

\subsection{Connection to identification via channels and other Ahlswede models in \cite{ahlswede_general_2008}}
\label{sec:Ahlswede_models}

The Boolean function computation problem is a generalization of the identification via channels problem. A $(n, N, \lambda_1, \lambda_2)$ identification (ID) code is a collection $\{(Q_i, D_i), i=1,\ldots, N\}$ with probability distributions $Q_i$ on $\mc X^n$ and $D_i\subseteq\mc Y^n$ such that the two types of error probability satisfy
\begin{align*}
&Q_iW(D_i^c)\leq \lambda_1, i=1,\ldots, N\\
&Q_iW(D_j)\leq \lambda_2, i, j =1,\ldots, N, i\neq j
\end{align*}
As discussed in Example \ref{example:S_constant} above, an $(n, N, \lambda_1, \lambda_2)$ ID code (assuming $N$ is a power of $2$ for simplicity) is equivalent to an $(n, \log N, \mc F, \lambda_1,\lambda_2)$ BFC code with  $\mc F:=\{f_j^{id}, j=1,\ldots, 2^m\}$ where $f_j^{id}$ is defined in  (\ref{eq:f_id}).

In \cite{ahlswede_general_2008}, Ahlswede proposed several models which generalize the identification via channel problem.  In all models,  there is a family of partitions $\Pi = \{\pi_j\}_j$ where each   $\pi_j=\{P_{j1},\ldots, P_{jr}\}$ is a partition of the messages $\mc M$, where we call $P_{ji}$ a member of the partition $\pi_j$. The communication problem is described as follows:
\begin{itemize}
\item the transmitter picks a message $i$ from $\mc M$, and transmits a codeword associated to the message $i$
\item upon receiving the channel output,  for any partition $\pi_j$ from the family $\Pi$, the decoder needs to decide which member of the partition $\pi_j$ contains the true message $i$
\end{itemize}
Identifying the set of messages $\mc M$ with $\{0,1\}^m$, our model is a special case of the above model, where each partition $\pi_j$ corresponds to a Boolean function $f_j$, which partitions $\mc M$ into two (i.e. $r=2$) disjoint sets, corresponding to $f^{-1}[1]$ and $f^{-1}[0]$, respectively. In particular,  the  $K$-identification problem~\cite{ahlswede_general_2008} (see Model 3 below)  is identical to our problem with the choice of functions $\wcf(K)$. We repeat the communication models discussed in~\cite{ahlswede_general_2008} (in its original numeration) below, and connect them to the BFC problem.
\begin{itemize}
\item Model 2 (identification): $\Pi_I=\{\pi_j: \pi_j=\{ \{j\}, \mc M\backslash \{j\}\}, j\in \mc M\}$. This corresponds to our problem with the choice  of functions to be $f_j^{id}$ defined in \eqref{eq:f_id}. Notice that we have
\begin{align*}
\{j\} = (f_j^{id})^{-1}[1]
\end{align*}
\item Model 3 ($K$-identification):  $\Pi_K=\{\pi_{\mc S}: \pi_{\mc S}=\{\mc S, \mc M\backslash S\}, |\mc S|= K, \mc S\subseteq \mc M\}$. This corresponds to our problem with the choice  of functions to be the set $\wcf(K)$. Here we have
\begin{align*}
\mc S = f^{-1}[1] \text{ for  } f\in \wcf(K)
\end{align*}
\item Model 4 (ranking): $\Pi_R = \{\pi_r: \pi_r = \{\{1,\ldots, r\}, \{r+1,\ldots, 2^m\}\}, 1\leq r\leq 2^m\}$. This corresponds to our problem with the choice of functions $f_r^{\text{rank}}$ defined in \eqref{eq:f_order}. Notice that
\begin{align*}
\{1,\ldots, r\} = (f_r^{\text{rank}})^{-1}[1]
\end{align*}
\item Model 5 (general binary questions) : $\Pi_B = \{\pi_{\mc A}: \pi_{\mc A}=\{\mc A, \mc M\backslash \mc A\}, \mc A\subset \mc M\}$. This corresponds to our problem with the set of functions $\{f_{\mc A}, \mc A\subset \mc M\}$, where $f_{\mc A}$ is the function whose preimage under $1$ is the set $\mc A$.
\item Model 6 ($t$-th bit): $\Pi_C = \{\pi_t: \pi_t =\{  \{\ve b\in\{0,1\}^m: \ve b_t=1\}, \{\ve b\in\{0,1\}^m: \ve b_t=0\}\} , t=1,\ldots, m \}$. This corresponds to our problem with the choice of functions to be $f_t^{bit}$ defined in \eqref{eq:f_tbit}. Notice that
\begin{align*}
 \{\ve b: \ve b_t=1\} = (f_t^{bit})^{-1}[1]
\end{align*}
\end{itemize}

It may be worth pointing out that  Shannon's original transmission problem (recover the transmitted message) is not covered in the BFC framework, but can be formulated in Ahlswede's partition framework (see Model 1 in \cite{ahlswede_general_2008}). Nevertheless, it can be formulated as a \textit{multiple Boolean functions computation} problem.  For example, if we extend Model 6 above, asking the the receiver to simultaneously compute $m$ functions $f_1^{bit},\ldots, f_m^{bit}$, then we effectively recover the original message. We only consider computing a single function in this paper, and leave multiple function computation for future work.

\section{Main results}
We state the main results of our paper in this section.  For Boolean functions with $m$ inputs, we consider the case when Hamming weight $S(m)$ takes integer values in the interval $[1, 2^m]$. The case  $\wcf(0)$ is trivial as this set only contains the constant function $0$ (hence the computation rate is infinite for every rate function). Notice that for the same reason $\wcf(2^m)$ is  trivial, whereas $\wcf(\leq 2^m)$ is not.

\begin{theorem}[A characterization of $\cbfc$ for $\wcf(S)$ and $\wcf(\leq S)$]
Assume the channel $W(\cdot | x^n)$ has Shannon capacity $C$.  Let $S:\mathbb N\rightarrow\mathbb N^+$ be a strictly increasing function and $S^{-1}$ the inverse function of $S$.  Let $\cbfc$ denote the computation capacity for the sequence of sets of functions $\{\wcf(S(m))\}_m$. We have the following results:
\begin{itemize}
\item (Small $S(m)$ regime) If $\frac{\log S(m)}{\log m}\rightarrow a$ for some $a\geq 0$ as $m\rightarrow \infty$, then $\cbfc\in[\frac{C}{1+2a} ,\frac{C}{1+a}]$ with the \rf $L(R,n):=2^{nR}$. 
\item  (Medium $S(m)$ regime) If $\frac{\log S(m)}{\log m}\rightarrow \infty$ as $m\rightarrow \infty$ and $\frac{\log S(m)}{m}\leq \alpha$ for some $\alpha\in(0,1)$, then $\cbfc\in[C/2, C]$ with the \rf $L(R,n):=S^{-1}(2^{nR})$. 
\item (Large $S(m)$ regime) If $\frac{\log S(m)}{m}\rightarrow 1$ as $m\rightarrow \infty$ and $S(m)\leq 2^{m-1},$ then $\cbfc=C$ with the \rf $L(R,n):=nR$.
\item The same results hold for the sequence of sets of functions $\{\wcf(\leq S(m))\}_m$. Moreover in this case, the result for the large $S(m)$ regime (i.e., $\frac{\log S(m)}{m} \rightarrow 1$) holds for all $S\leq 2^{m}$.
\end{itemize}
\label{thm:capacity_bfc}
\end{theorem}

The above theorem completely characterizes the rate function for the BFC problem for the function classes $\mc F(S)$ and $\mc F(\leq S)$, respectively. Furthermore, the computation capacity is determined within upper and lower bounds which differ by a factor of at most $2$ for all choices of $S$. Now we specialize the result in Theorem \ref{thm:capacity_bfc} to several choices of $S(m)$, all of them are in the small or medium regime.

\begin{proposition}[Results for certain choices of $S(m)$]
Assume the channel $W(\cdot | x^n)$ has a Shannon capacity $C$, and let $c>0$ be a constant independent from $m$ and $n$.  Let  $\cbfc$ denote the computation capacity for $\{\wcf(S(m))\}_m$ with a \rf specified below. We have the following statements (S for small, M for medium):
\begin{itemize}
\item (1, S) for $S=c$ (a positive integer), $\cbfc = C$ with the \rf $L(R,n)=2^{Rn}$
\item (2, S) for $S=c(\log m)^b$ with $b>0$, $\cbfc = C$ with the \rf $L(R,n)=2^{Rn}$
\item (3, S) for $S=cm^{\beta}$ with $\beta>0$, $\cbfc \in[\frac{C}{1+2\beta}, \frac{C}{1+\beta}]$ with the \rf $L(R,n)=2^{Rn}$
\item (4, M) for $S=cm^{(\log m)^b}$ with $b>0$, $\cbfc \in [C/2, C]$ with the \rf $L(R,n)=2^{(Rn)^{1/(b+1)}}$
\item (5, M) for $S=c2^{m^{1/b}}$ with $b> 1$, $\cbfc \in[C/2,C]$ with the \rf $L(R,n)=(Rn)^b$
\item (6, M) for $S=c2^{m/\log m}$, $\cbfc \in[C/2, C]$ with the \rf $L(R,n)=Rn\log n$
\item (7, M) for $S=c2^{\gamma m}$ with $\gamma\in(0,1)$, $\cbfc \in[\max\{\frac{C}{2\gamma}, C\},  \frac{C}{\gamma}]$ with the \rf $L(R,n)=Rn$
\end{itemize}
The same results hold for the sequence $\{\wcf(\leq S(m))\}_m$.
\label{prop:special_S}
\end{proposition}

Theorem \ref{thm:capacity_bfc} (and Proposition \ref{prop:special_S}) show that the \rf of the BFC problem crucially depends on the size of $S(m)$.  We call $S(m)$ \textit{small} if $\log S(m)=O(\log m)$, in which case the achievable $m$ scales as $O(2^n)$, the same as in the identification problem. If $\log S(m)$ grows faster than $O(\log m)$ but $S(m)$ grows slower than $O(2^m)$, we call $S(m)$ \textit{medium}, and the \rf depends on the specific form of $S(m)$. We call $S(m)$ \textit{large} if $\log S(m)/m\rightarrow 1$. Our result  shows that the communication efficiency of the BFC code for $\wcf(S)$ with a large $S(m)$ is the same as that of the Shannon problem, namely $m$ scales as $nC$ asymptotically. As discussed in Remark \ref{remark:symmetry}, without loss of generality, we only need to consider $S\in(0,\frac{1}{2}2^{m}]$ for the function class $\wcf(S)$ due to symmetry.  For the function class $\wcf(\leq S)$, the result in the small and the medium regime is the same as for $\wcf(S)$, however the large regime extends to $S=2^m$. We show the results pictorially in Figure \ref{fig}.

\begin{figure*}[!htb]
    \centering
    \begin{subfigure}[t]{0.5\textwidth}
        \centering
        \includegraphics[scale=0.35]{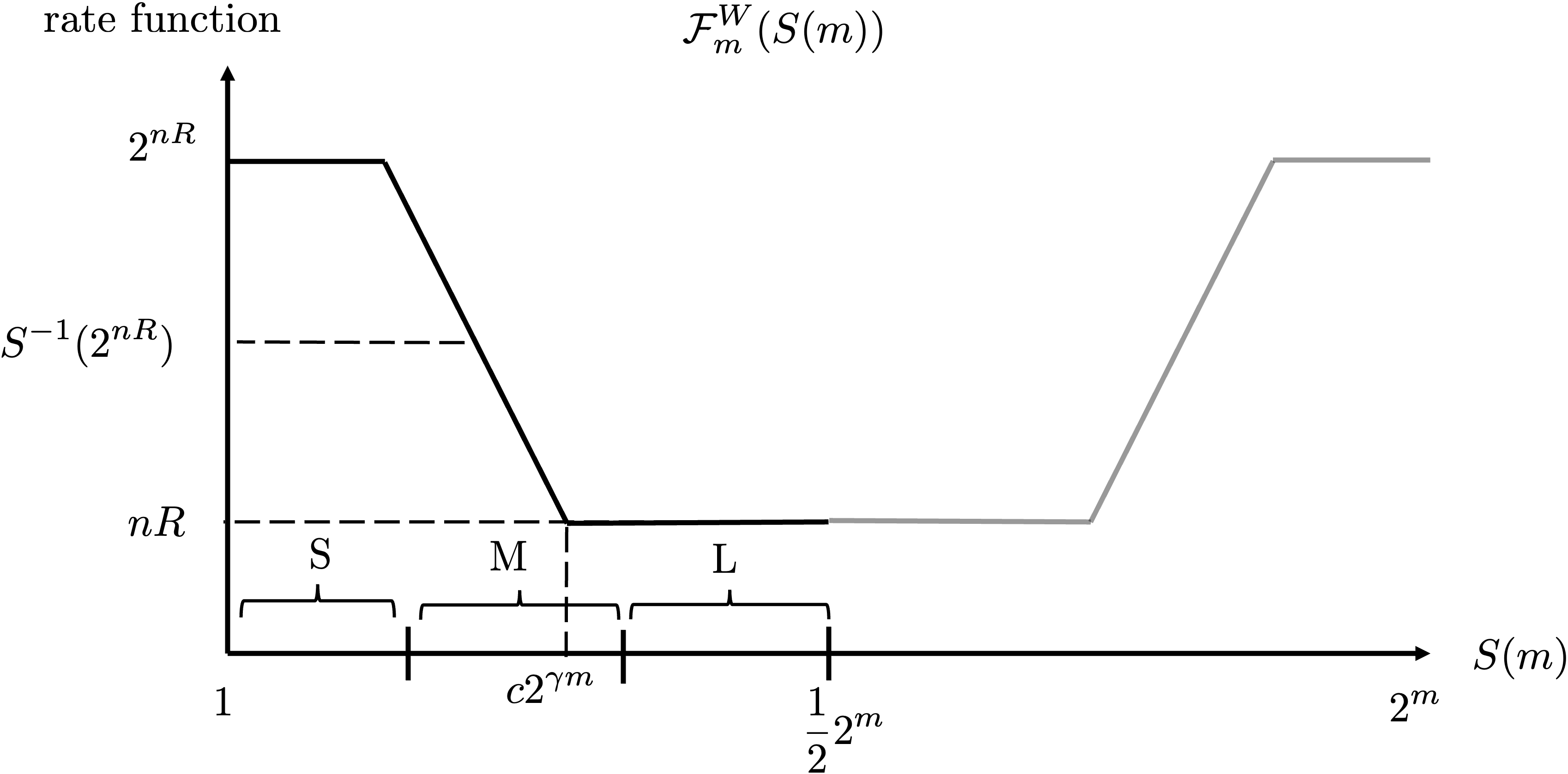}
        \caption{rate function vs. $S(m)$ for the class $\wcf(S(m))$}
    \end{subfigure}%
    \begin{subfigure}[t]{0.5\textwidth}
        \centering
        \includegraphics[scale=0.35]{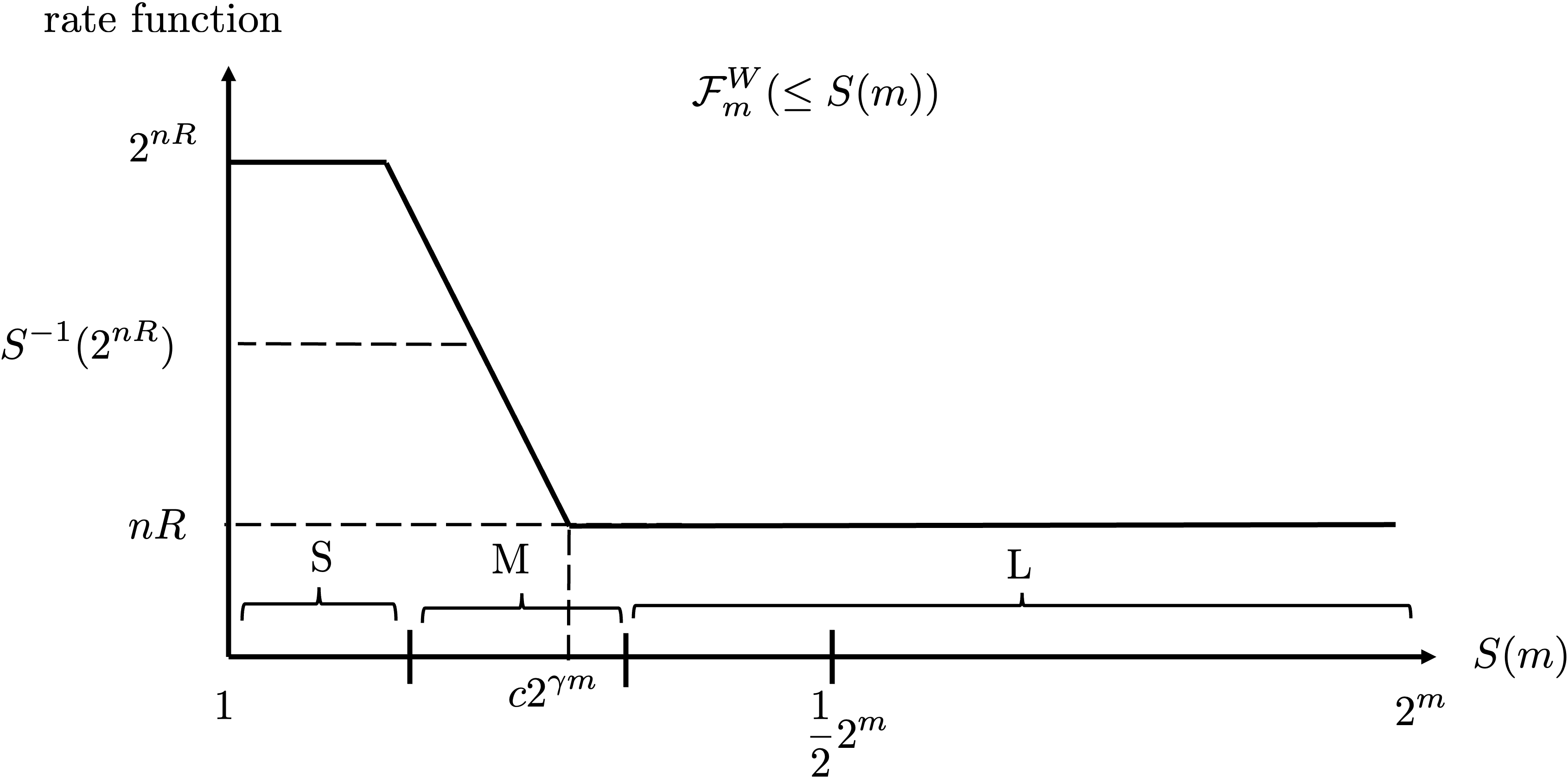}
        \caption{rate function vs. $S(m)$ for the class $\wcf(\leq S(m))$}
    \end{subfigure}
    \caption{The two figures qualitatively depict the behaviour of rate function for different $S(m)$. Subfigure (a) shows the rate function for the function class $\wcf(S(m))$. In the small $S(m)$ regime (S), the rate function is $2^{nR}$. In the medium $S(m)$ regime (M), the rate function varies as $S^{-1}(2^{nR})$. In particular, when $S(m)=c2^{\gamma m}$ for some $\gamma \in(0,1)$, the rate function ``drops" to $nR$. As $S(m)$ continues to grow into the large regime (L), the rate function stays at $nR$. The capacity results (hence the rate function) for $\wcf(S(m))$ are symmetric around the value $S=2^{m-1}$. Subfigure (b) shows the rate function for the function class $\wcf(\leq S(m))$, whose capacity result is the same as $\wcf(S(m))$ for $S(m)\leq 2^{m-1}$. However in this case, the rate is not symmetric, and  the computation capacity is $C$ with the rate function $nR$ for all $S\geq 2^{m-1}$.}
    \label{fig}
\end{figure*}

Theorem \ref{thm:capacity_bfc} is derived from the following two theorems, proving an achievability and a converse result, respectively.


\begin{theorem}[Achievability for $\wcf(\leq S)$]
Assume the channel $W(\cdot | x^n)$ has Shannon capacity $C>0$. Let $S:\mathbb N\rightarrow\mathbb N$ be a function, and define
\begin{align}
L_S(R,n):=\max\left\{m\in\mathbb N: mS(m)^2\leq 2^{nR}\right\}.
\label{eq:L_S}
\end{align}
Then $C$ is an achievable computation rate for the function class  $\wcf(\leq S(m))$ with the \rf $L_S(R,n)$.
\label{thm:general_achievability}
\end{theorem}

\begin{theorem}[Converse for $\wcf(S)$]
Assume the channel $W(\cdot | x^n)$ has Shannon capacity $C>0$. Let $S:\mathbb N\rightarrow\mathbb N^+$ satisfy $1\leq S(m)\leq c2^{\gamma m}$ for some $\gamma\in [0,1)$ and $c>0$.  If there exists an $(n,m, \wcf(S(m)),\lambda_1,\lambda_2)$ BFC code with $\lambda_1+\lambda_2<1$, then for every $\delta>0$ and sufficiently large $m,n$, it must hold that
\begin{align*}
mS(m)\leq 2^{n(C+\delta)}.
\end{align*}
Furthermore if $S(m)=c2^{\gamma (m)m}$ where $\gamma(m)\rightarrow 1$ for some $c\in(0,1/2]$, it must hold $m\leq n(C+\delta)$ for sufficiently large $m, n$.
\label{thm:general_converse}
\end{theorem}

\begin{remark}
The converse results in the above theorem hold for $\wcf(S)$, and it follows directly that the stated results also hold for $\wcf(\leq S)$ as it is a larger class that contains $\wcf(S)$. Similarly, the achievability results in Theorem \ref{thm:general_achievability}, stated for the function class of the form $\wcf (\leq S)$, also hold for the class $\wcf(S)$, as the former include the latter.
\label{remark:converse_achievable}
\end{remark}

The results in  Theorem \ref{thm:capacity_bfc} and Proposition \ref{prop:special_S} are summarized in Table \ref{table:summary}, where in the first row we show the general results, followed by specific choices of $S(m)$. The scaling law of $m$ with respect to $n$ is shown in the last column. In the last row of the table we include the result of  Shannon's transmission problem as a reference, though noting that it is not an instance of the BFC problem.
\begin{table*}[tbh!]
\centering
\begin{tabular}{l c c c}
\toprule
Hamming weight $S(m)$  &  \makecell{Achievability \\ (lower bound on $m$) }  & \makecell{Converse \\ (lower bound on $m$)} &  Asymptotic expression  of $m$ \\
\midrule
\quad (general result) & $mS(m)^2\leq 2^{nC}$ & $mS(m)\leq 2^{nC}$ &   $\Theta(2^{n})$ or $\Theta(S^{-1}(2^{n}))$ or $\Theta(n)$ \\
\midrule
(1) \quad $c$ (constant) & $2^{Cn}$ (\cite{ahlswede_identification_1989}  for $c=1$) & $2^{Cn}$(\cite{ahlswede_identification_1989} for $c=1$)  & $\Theta(2^n)$, exponential  \\
\midrule
(2) \quad $c(\log m)^b, b>0$  & $2^{Cn}$   & $2^{Cn}$ & $\Theta(2^n)$, exponential  \\
\midrule
(3) \quad $cm^{\beta}, \beta>0$  & $2^{\frac{C}{1+2\beta}n}$ (\cite{ahlswede_general_2008}) for $c=1$) & $2^{\frac{C}{1+\beta}n}$ (\cite{ahlswede_general_2008} for $c=1$) & $\Theta(2^n)$, exponential \\
\midrule
(4) \quad $cm^{(\log m)^b}, b>0$  & $2^{(Cn/2)^{1/(1+b)}}$ & $2^{(Cn)^{1/(1+b)}}$ &$\Theta\left(2^{n^{1/(1+b)}}\right) $, sub-exponential\\
\midrule
(5) \quad $c2^{m^{1/b}}, b>1$ &  $(\frac{C}{2}n)^b$ & $(Cn)^{b}$ & $\Theta(n^{b})$, polynomial \\
\midrule
(6) \quad $c2^{m/\log m}$ & $\frac{C}{2}n\log n$  & $C n\log n$ & $\Theta(n\log n)$, quasi-linear \\
\midrule
(7) \quad $c2^{\gamma m}, \gamma\in(0,1)$ & $\max\{Cn, \frac{C}{2\gamma}n\}$ & $\frac{C}{\gamma}n$ & $\Theta(n)$, linear \\
\midrule
Shannon problem & $Cn$ & $Cn$ & $\Theta(n)$\\
\bottomrule
\end{tabular}
\caption{A summary of the results in Theorem \ref{thm:capacity_bfc} and Proposition \ref{prop:special_S}. Here $m$ denotes the length of the message and $n$ is the number of channel uses. In the last row we include the Shannon problem for reference.}
\label{table:summary}
\end{table*}


We have the following remarks:
\begin{itemize}
\item The general results in Theorem \ref{thm:capacity_bfc}, \ref{thm:general_achievability} and \ref{thm:general_converse} are new. Results related to Cases 1), 3) in Proposition \ref{prop:special_S}, and to the large $S$ regime in Theorem \ref{thm:capacity_bfc} already  appeared in the literature, as discussed below.
\item As discussed in Section \ref{sec:Ahlswede_models}, in Case (1) with the choice $c=1$, the BCF problem with $\wcf(1)$ is equivalent to the identification via channels problem studied in \cite{ahlswede_identification_1989} where the double exponential capacity result was established.  The same result holds if the Hamming weight is a constant larger than $1$.
\item For Case (3), the choice $S=m^{\beta}$ has been studied in \cite[Section 3]{ahlswede_general_2008} under the name $K$-identification (see Model 3 in Section \ref{sec:Ahlswede_models}) where the Hamming weight $S$ (denoted by $K$ in \cite{ahlswede_general_2008}) is parametrized as $2^{\kappa  n}$ (i.e. by $n$ instead of $m$). It was shown that $R=C-2\kappa$ (corresponding to $m=2^{n(C-2\kappa)}$ in our case) is achievable and a converse result of $R\leq C-\kappa$ (corresponding to $m\leq 2^{n(C-\kappa)}$ in our case) is established. Identifying $2^{\kappa n}$ with $m^{\beta}$, it can be checked that our result matches that in \cite{ahlswede_general_2008}. For example,  our achievable rate $m=2^{\frac{C}{1+2\beta}n}$ is matched with the result $2^{n(C-2\kappa)}$ in~\cite{ahlswede_general_2008} by letting $2^{\kappa n}=m^{\beta}$:
\begin{align*}
2^{n(C-2\kappa)}&=2^{nC}\cdot (2^{\kappa n})^{-2}= 2^{nC}\cdot (m^\beta)^{-2}\\
&=2^{nC}\cdot 2^{\frac{-2\beta }{1+2\beta}Cn}=2^{\frac{C}{1+2\beta}n}
\end{align*}
\item Several communication problems are discussed in \cite[Section 4]{ahlswede_general_2008} whose capacity is equal to the ordinary Shannon capacity, namely, $m\approx Cn$ asymptotically. These problems include the ranking problem in Example \ref{example:ranking} (Model 4) with the choice $S=2^m$, and the $t$-th bit identification problem in Example \ref{example:t_bit} (Model 6). It is discussed in Section \ref{sec:example_BFC} that the former problem corresponds to a BFC problem with the set $\wcf(\leq 2^m)$ and the latter problem corresponds to a BFC problem with the set $\wcf(\frac{1}{2}2^m)$. Both examples correspond to the large $S$ regime in Theorem \ref{thm:capacity_bfc}.
 \end{itemize}


\subsection{Proof of Theorem \ref{thm:capacity_bfc} and Proposition \ref{prop:special_S}}

We will prove the general achievability (Theorem \ref{thm:general_achievability}) and  converse results (Theorem \ref{thm:general_converse}) in the next section. Here we give a proof of Theorem \ref{thm:capacity_bfc} based on these two results.

\begin{proof}[Proof of Theorem \ref{thm:capacity_bfc}] \textbf{(Small $S$ regime.)} We first consider the case when $\frac{\log S(m)}{\log m}\rightarrow a$ for some $a> 0$, which means that for any $\varepsilon>0$, we have $(1+\varepsilon)a\log m\geq \log S(m)\geq (1-\varepsilon)a\log m$ for large enough $m$.    The achievability result in Theorem \ref{thm:general_achievability} states that  for every $\eta>0$ and sufficiently large $n$, there exists an $(n,m,\wcf(\leq S(m))$ code as long as $m$ satisfies $\log m+2\log S(m)\leq n(C-\eta)$ for any $\eta>0$. In this case, this implies that there exists an $(n,m,\wcf(\leq S(m))$ code if  $m$ satisfies (for large enough $n$ and $m$)
\begin{align*}
\log m+2(1+\varepsilon)a\log m\leq n(C-\eta)
\end{align*}
which is equivalent to
\begin{align*}
m\leq 2^{\frac{(C-\eta)}{1+2(1+\varepsilon)a}n}
\end{align*}
Furthermore, it is easy to see that for any desired small $\eta'>0$, we can choose $\eta$ and $\varepsilon$ sufficiently small such that  $\frac{(C-\eta)}{1+2(1+\varepsilon)a}\geq \frac{C}{1+2a}-\eta'$. Indeed, notice that $
\frac{C}{1+2a}-\frac{(C-\eta)}{1+2(1+\varepsilon)a}=\frac{\eta(1+2a)+2aC\varepsilon}{(1+2a)(1+2(1+\varepsilon)a)}$ which approaches $0$ as $\eta, \varepsilon\rightarrow 0$. This shows there exists an $(n,m,\wcf(\leq S(m),\lambda_1,\lambda_2)$ code as long as $m$ satisfies $m\leq 2^{(C/(1+2a)-\eta')n}$, meaning $C/(1+2a)$ is achievable with the \rf $L(R,n)=2^{nR}$. 

If $\frac{\log S(m)}{\log m}\rightarrow 0$,  for all $\epsilon$, we have  $\log S(m)\leq \epsilon\log m$ for $m$ large enough. Theorem \ref{thm:general_achievability} gives achievability whenever $\log m+2\log S(m)\leq (1+2\epsilon)\log m \leq n(C-\eta)$. Thus every rate below $C/(1+2\epsilon)$ is achievable, and letting $\epsilon$ approach $0$ gives $C$. Since $\wcf(S(m))$ is contained in $\wcf(\leq S(m))$, the achievability result also holds for $\wcf(S(m))$.

Now we give an upper bound on $\cbfc$ in this case. If $\frac{\log S(m)}{\log m}\rightarrow a>0$, we  have $S(m)\leq m^{a+\epsilon}$ for all $\epsilon>0$ for large enough $m$, hence $S(m)\leq 2^{\gamma m}$ for some $\gamma\in(0,1)$ and large enough $m$. Theorem \ref{thm:general_converse} states that in this case, given an $(m,n,\wcf(S(m)), \lambda_1,\lambda_2))$ code, for sufficiently large $m,n$, it must hold that $\log m+\log S(m)\leq n(C+\delta)$ for every $\delta>0$, which implies it must hold $\log m+(1-\varepsilon)a\log m\leq n(C+\delta)$ for large enough $m,n$, or equivalently
\begin{align*}
m\leq 2^{\frac{(C+\delta)}{1+(1-\varepsilon)a}n}
\end{align*}
Similarly, as in the direct case above, for any desired $\delta'>0$, we can choose $\eta$ and $\varepsilon$ sufficiently small such that $\frac{(C+\delta)}{1+(1-\varepsilon)a}\leq \frac{C}{1+a}+\delta'$. This shows that $m$ must satisfy $m\leq 2^{(C/(1+a)+\delta')n}$ for large enough $m,n$, meaning $\cbfc\leq C/(1+a)$ with the \rf $L(R,n)=2^{nR}$.  In the case when $\frac{\log S(m)}{\log m}\rightarrow 0$, as $S(m)>0$, Theorem \ref{thm:general_converse} implies $\log m\leq n(C+\delta)$, hence $\cbfc\leq C$. As mentioned in Remark \ref{remark:converse_achievable}, as $\wcf(\leq S(m))$ contains the set $\wcf(S(m))$, the above converse result also applies to any $(m,n,\wcf(\leq S(m)), \lambda_1,\lambda_2))$ code.

\textbf{(Medium $S$ regime.)} Now consider the second case when $\frac{\log S(m)}{\log m}\rightarrow \infty$  as $m\rightarrow\infty$ and $\frac{\log S(m)}{m}\leq \alpha\in (0,1)$. In this case for any $\varepsilon>0$, we have $\log m\leq  \varepsilon \log S(m)$ for large enough $m$. Theorem \ref{thm:general_achievability} implies that there exists an $(n,m,\wcf(\leq S(m))$ code if  $m$ satisfies (for large enough $n$)
\begin{align*}
\varepsilon\log S(m)+2\log S(m)\leq n(C-\eta)
\end{align*}
which is equivalent to
\begin{align*}
m\leq S^{-1}\left(2^{\frac{C-\eta}{2+\varepsilon}n}\right).
\end{align*}
For desired small $\eta'$, we can choose $\eta$ and $\varepsilon$ sufficiently small such that $(C-\eta)/(2+\varepsilon)\geq C/2-\eta'$. Furthermore, because $S^{-1}$  is strictly increasing (as $S$ is strictly increasing), we have shown that there exists an $(n,m,\wcf(\leq S(m))$ code as long as $m$ satisfies $m\leq S^{-1}(2^{(C/2-\eta')n})$, which shows $C/2$ is an achievable computation rate with the \rf $L(R,n)=S^{-1}(2^{nR})$.  Since $\wcf(S(m))$ is contained in $\wcf(\leq S(m))$, the achievability result also holds for $\wcf(S(m))$.

To give an upper bound on $\cbfc$ in this case, notice $\frac{\log S(m)}{m}\leq \alpha$ implies $S(m)\leq 2^{\alpha m}$ for some $\alpha\in(0,1)$. Theorem \ref{thm:general_converse} states that in this case, for sufficiently large $m,n$, it must hold that $\log m+\log S(m)\leq n(C+\delta)$ for every $\delta>0$, which implies it must hold $\log S(m)\leq n(C+\delta)$ for large enough $m,n$, or equivalently
\begin{align*}
m\leq S^{-1}(2^{(C+\delta)n})
\end{align*}
This shows that in this case $\cbfc\leq C$ with the \rf $L(R,n)=S^{-1}(2^{nR})$. As $\wcf(\leq S(m))$ contains the set $\wcf(S(m))$, the above converse result also applies to any $(m,n,\wcf(\leq S(m)), \lambda_1,\lambda_2))$ code. 

\textbf{(Large $S$ regime.)} In the case when $\frac{\log S(m)}{m}\rightarrow 1$ with $S(m)\leq 2^{m-1}$, we can rewrite 
 $S(m)=c2^{\gamma(m)m}$ for some $c\in (0,1/2]$ and some $\gamma(m)\rightarrow 1$ as $m\rightarrow \infty$. Indeed, we could choose $\gamma(m)=\frac{\log (S(m)/c)}{m}$ with $c=1/2$, and it is easy to check that $\gamma(m)\rightarrow 1$ in this case. For this choice of $S(m)$, Theorem \ref{thm:general_converse}  states that $\cbfc\leq C$ with the \rf $L(R,n)=nR$ for any $(m,n,\wcf(c2^{\gamma(m)m}), \lambda_1,\lambda_2)$ code with $c\in(0,1/2]$. The converse also applies to the class $\wcf(\leq c2^{\gamma(m)m})$ with $c\in(0,1/2]$. On the other hand, $C$ is also an achievable computation rate.  To achieve this rate, we could simply use the original error correction codes (described in \eqref{eq:good_channel_code}) and use a ``separation approach", namely to let the receiver recover the message $i$ reliably as in the classical transmission scheme, and then compute $f_j(i)$ at the receiver.

To cover the converse result for the remaining case: $\wcf(\leq S(m))$ where $\frac{\log S(m)}{m}\rightarrow 1$ and $2^{m-1}\leq S(m)\leq 2^m$, we notice that similarly we can rewrite $S(m)$ as $S(m)=c2^{\gamma(m)m}$ for some $c\in[1/2,1]$. Now we make a simple observation: $\wcf(\leq \frac{1}{2}2^{\gamma(m)m}) \subseteq\wcf(\leq c2^{\gamma(m)m})$ where $c\in[1/2,1]$, which implies any converse result on the former class should also apply to the latter class, and we already know that for the former case, $\cbfc\leq C$ with the \rf $nR$. This concludes the proof.
\end{proof}

Now we prove Proposition \ref{prop:special_S}. The purpose is to provide a more explicit characterization of the \rf for specific classes of $S(m)$.

\begin{proof}[Proof of Proposition \ref{prop:special_S}] 

{\bf Case 1 (constant $S$).} When $S(m)=c$ for some constant $c$,  we have $\log c/\log m\rightarrow 0$ as $m\rightarrow \infty$. Theorem \ref{thm:capacity_bfc} states that in this case $\cbfc=C$ with the \rf $L(R,n)=2^{nR}$.

{\bf Case 2 ($S=c(\log m)^b$).} When $S=c(\log m)^b$ for some $c> 0, b\geq 0$, we have $\log S(m)/\log m = \frac{\log c+b\log(\log m)}{\log m}\rightarrow 0$.  Theorem \ref{thm:capacity_bfc} states that in this case $\cbfc=C$ with the \rf $L(R,n)=2^{nR}$.

{\bf Case 3 ($S=c m^\beta$).}  When $S=cm^{\beta}$ for some constant $c,\beta>0$  not depending on $m, n$, we have 
\begin{align*}
\log S(m)/\log m = (\log c+\beta\log m)/\log m\rightarrow \beta.
\end{align*}
Theorem \ref{thm:capacity_bfc} states that in this case $\cbfc\in [\frac{C}{1+2\beta}, \frac{C}{1+\beta}]$ with the \rf $L(R,n)=2^{nR}$.

\textbf{Case 4 ($S=cm^{(\log m)^b}$ for $b>0$).}   It is easy to check that $S(m)=cm^{(\log m)^b}$ is strictly increasing for $b>0$, and we have $S^{-1}(y)=2^{(\log(y/c))^{1/(b+1)}}$. Notice that in this case it holds $\log S(m)/\log m = (\log c+ (\log m)^{b+1})/\log m\rightarrow \infty$, so Theorem \ref{thm:capacity_bfc} states that $\cbfc\in [C/2,C]$ with the \rf $S^{-1}(2^{nR})=2^{(nR-\log c)^{1/(1+b)}}$. Now define the rate function $L'(R,n):=2^{(nR)^{1/(1+b)}}$. It is easy to check that $L'(R,n)$ is a valid rate function in the sense of Definition \ref{def:valid_rate_function}. Furthermore, $L'(R,n)$ is equivalent to $S^{-1}(2^{nR})$ in the sense of Definition \ref{def:equivalent_rate_function}. Then Lemma \ref{lemma:equivalence} states that the rate $\cbfc$ is also achievable with $L'$, which proves the claimed result.

\textbf{Case 5 ($S=c2^{m^{1/b}}$ for $b>1$).}  It is easy to check that $S(m)=c2^{m^{1/b}}$ is strictly increasing for $b>1$, and we have $S^{-1}(y)=(\log (y/c))^b$. In this case $\log S(m)/\log m=(\log c+m^{1/b})/\log m\rightarrow \infty$, so Theorem \ref{thm:capacity_bfc} states that $\cbfc\in[C/2,C]$ with the \rf $S^{-1}(2^{nR})=(nR-\log c)^b$. As $L'(R,n):=(nR)^b$ is a valid rate function which is also equivalent to $S^{-1}(2^{nR})$, we have the claimed result.

{\bf Case 6 ($S=c2^{m/\log m}$)}. Since $S(m)=c2^{m/\log m}$ is strictly increasing for $m\geq 3$, we could apply Theorem \ref{thm:capacity_bfc} to this case. However, the inverse function $S^{-1}$ in this case does not have an explicit expression. We show in Appendix \ref{appendix:equivlent_rate_functions} that $S^{-1}(2^{nR})$ and is equivalent to the rate function $L(R,n):=n\log n$ and invoke Lemma \ref{lemma:equivalence} for the claim.

{\bf Case 7 ($S=c2^{\gamma m}$ for $\gamma\in(0,1)$}).  $S(m)=c2^{\gamma m}$ is strictly increasing and we have $S^{-1}(y)=\frac{1}{\gamma}\log(y/c)$. In this case $\log S(m)/\log m = (\log c+\gamma m)/\log m\rightarrow \infty$, so Theorem \ref{thm:capacity_bfc} states that $\cbfc\in[C/2,C]$ with the \rf $S^{-1}(2^{nR})=\frac{1}{\gamma}(nR-\log c)$. As $L'(R,n):=\frac{nR}{\gamma}$ is equivalent to $S^{-1}(2^{nR})$, we proved that $\cbfc\in[C/2,C]$ with the \rf $nR/\gamma$. Equivalently, the result can be stated as $\cbfc\in [C/(2\gamma), C/\gamma]$ with the \rf $L(R,n):=nR$.

On the other hand, the computation rate $C$ for the \rf $L(R,n)=Rn$ is also achievable. To achieve this rate, we use a ``separation approach", namely to let the receiver recover the message $i$ reliably as in the classical Shannon's transmission problem, and then compute $f_j(i)$ at the receiver. In this sense, the encoder/decoder construction at the beginning of this section is not needed for this result.  Overall, we have $\cbfc\in[\max\{C, \frac{C}{2\gamma}\}, C/\gamma]$ with the \rf $L(R,n)=Rn$.
\end{proof}

\section{Proof of Theorem \ref{thm:general_achievability} and  Theorem\ref{thm:general_converse}}
Now we present the main proofs in this paper, namely the proofs of Theorem \ref{thm:general_achievability} and \ref{thm:general_converse}.

\subsection{Proof of achievability}
\label{sec:proof_achievability}

In this section we prove the achievability result in Theorem \ref{thm:general_achievability}. Central to the proof is the following result, stated in~\cite[Proposition 1]{ahlswede_identification_1989}. The following proposition is a slightly modified version of \cite[Proposition 1]{ahlswede_identification_1989}. 

\begin{proposition}[Maximal code, modified from \cite{ahlswede_identification_1989}, Proposition 1]
\label{prop:maximal_code}
Let $\mc Z$ be a finite set with cardinality $|\mc Z|>6$  and let $\lambda\in(0,1/2)$ be given. Let $\epsilon<\frac{1}{6}$ and $\epsilon |\mc Z|\geq 1$. Then a family $A_1,\ldots, A_N$ of subsets of $\mc Z$ satisfying the following properties exists 
\begin{align*}
|A_i|&= M':=\lceil\epsilon |\mc Z|\rceil\geq 1, i=1,\ldots, N\\
|A_i \cap A_j|&<\lambda M', i,j=1,\ldots, N, i\neq j
\end{align*}
and
\begin{align*}
N\geq H(\lambda,M')M'^{-1}\left( \frac{1-\epsilon}{2\epsilon}\right)^{\lceil\lambda M'\rceil}
\end{align*}
where $H(\lambda,M'):={M'\choose \lceil\lambda M'\rceil}^{-1}$.  
\end{proposition}

The above result gives a tighter lower bound on $N$ than that stated in \cite[Proposition 1]{ahlswede_identification_1989}, which is the essential difference that is needed to prove Theorem \ref{thm:general_achievability}.  The proof is almost the same as the proof of \cite[Proposition 1]{ahlswede_identification_1989}, and is included in the Appendix for completeness.

\textbf{Construction of BFC codes.} We use the fact that for a channel with the (Shannon) capacity $C>0$, for all  $\delta\in(0,1/2)$, $\xi>0$, and  all sufficiently large $n$, there exists an $n$-length transmission code $\{(\ve u_i, \mc E_i)| i=1,\ldots, M\}$, where all $\ve u_i\in\mathcal X^n$ are different, with the following property
\begin{align}
&\mc E_i\cap\mc E_j=\emptyset, i\neq j, \nonumber\\
&W(\mc E_i^c|\ve u_i)\leq \delta, \nonumber\\
&M= 2^{n(C-\xi)}.  \label{eq:good_channel_code}
\end{align}

Fix a small error probability $\delta\in(0,1/2)$,  let $\mc Z=\{\ve u_1,\ldots, \ve u_M\}$ be the set of codewords from a transmission code satisfying the properties in \eqref{eq:good_channel_code}\footnote{Here we assume $2^{n(C-\xi)}$ is an integer for simplicity of the analysis. Without this assumption, we can assume that $M$ is some integer that satisfies $2^{n(C-\xi_1)}\leq M\leq 2^{n(C-\xi_2)}$ for arbitrarily small $\xi_1$ and $\xi_2$, and the same result can be proved by modifying the proof steps accordingly.}.   Then, there exist $N$ subsets of $\mc Z$ with the property specified in Proposition \ref{prop:maximal_code}. We use $A_i, i\in \{0,1\}^m$ (indexed by binary sequences of length $m$) to denote the first $2^m$ of the $N$ subsets, where $m$ will be chosen to be an integer such that
\begin{align}
2^m\leq H(\lambda,M')M'^{-1}\left( \frac{1-\epsilon}{2\epsilon}\right)^{\lceil\lambda M'\rceil}=:\tilde N
\label{eq:m}
\end{align}
where $M'=\lceil\epsilon M\rceil$. Proposition \ref{prop:maximal_code} guarantees the existence of $\{A_i\}_{i\in\{0,1\}^m}$ if \eqref{eq:m} is satisfied. Also notice that $\epsilon$ should satisfy 
\begin{align}
\frac{1}{6}>\epsilon\geq  \frac{1}{|\mc Z|}=2^{-n(C-\xi)}
\label{eq:epsilon_const}
\end{align}
by the assumption in Proposition \ref{prop:maximal_code}. 

The encoder and the decoder of the BFC code are defined as follows.

\textit{Encoder:} Define for every $i\in\{0,1\}^m$, the stochastic encoder $Q_i$ to be the uniform distribution over $A_i$:
\begin{align}
Q_i(\ve x)=\ve 1_{\ve x\in A_i}\frac{1}{|A_i|} 
\label{eq:encoder}
\end{align}

\textit{Decoder: }The decoder set for the function $f_j$ is defined to be
\begin{align}
\mathcal D_j &:= \bigcup_{i\in f_j^{-1}[1]}\bigcup_{k:  \ve u_k\in A_i} \mc E_k\nonumber\\
&=\bigcup_{k: \ve u_k\in B_j}\mc E_k \label{eq:Dj_def}
\end{align}
where we define $B_j:=\bigcup_{\ell\in f_j^{-1}[1]}A_\ell$. 


The following lemma serves as an intermediate result to characterize the error probability of the above coding scheme.

\begin{lemma}[Error probability]
Fix $S\in\intset{2^m}$. For all $\delta,\xi\, \lambda \in(0,1/2)$ and all sufficiently large $n$, if $m$, $\epsilon$ and $\lambda$ satisfy \eqref{eq:m} and \eqref{eq:epsilon_const}, then the  $(n,m,\wcf(\leq S), \lambda_1, \lambda_2)$  BCF code  constructed above satisfies 
\begin{align*}
\lambda_1\leq \delta, \lambda_2\leq S\lambda+\delta.
\end{align*}
\label{lemma:error_probability}
\end{lemma}
The lemma is proved in Appendix \ref{appendix:error_probability}. It shows that the false negative error $\lambda_1$ can be made arbitrarily small as the result holds for all $\delta\in(0,1/2)$.   In the following, we prove Theorem \ref{thm:general_achievability} by analysing the false positive error $\lambda_2$.


\begin{proof}[Proof of Theorem \ref{thm:general_achievability}]
Fix an arbitrary $\eta\in(0,C)$, choose $0<\xi<\eta/2$, and for any given $\delta>0$, we  make  the following choices
\begin{align}
\lambda=\frac{\delta}{S}, \quad \epsilon=\frac{\lambda}{4e}, \quad m=L_S(C-\eta,n)  
\label{eq:choice}
\end{align}
Where $L_S$ is defined in \eqref{eq:L_S}. Now we check that with the above choices, conditions in \eqref{eq:epsilon_const} and \eqref{eq:m} are satisfied.

We start with \eqref{eq:epsilon_const}. To check the upper bound for $\epsilon$, notice that
\begin{align*}
\epsilon= \frac{\lambda}{4e}=\frac{\delta}{4Se}\leq \frac{\delta}{4e}< \frac{1}{8e}
\end{align*}
where we used the fact that $S\geq 1$ and $\delta<1/2$.  Now we check that the lower bound in \eqref{eq:epsilon_const} holds for sufficiently large $n$. With the choice of $m$ in \eqref{eq:choice}, we have $m\leq S^{-2}2^{n(C-\eta)}$ hence 
\begin{align}
S\leq 2^{n(C-\eta)/2}m^{-1/2}\leq 2^{n(C-\eta)/2}
\label{eq:S_upper_bound}
\end{align}
as $m\geq 1$. The lower bound in \eqref{eq:epsilon_const} holds for sufficiently large $n$ by substituting the choices of $\epsilon$ and by   observing 
\begin{align*}
\epsilon 2^{n(C-\xi)}=\frac{\lambda}{4e} 2^{n(C-\xi)}&=\frac{\delta}{4Se}2^{n(C-\xi)}\\
&\geq \frac{\delta}{4e}2^{n(C/2-\xi+\eta/2)}\\
&\geq \frac{\delta}{4e}2^{nC/2}
\end{align*}
where the second last step follows from the upper bound of $S$ in \eqref{eq:S_upper_bound}, and the last step holds as we choose $\eta>2\xi$.  Since $\frac{\delta}{4e}2^{nC/2}\geq 1$ for large enough $n$, the lower bound in  \eqref{eq:epsilon_const} holds for sufficiently large $n$.

Now we show the choice of $m$ also satisfies the condition in \eqref{eq:m} by lower bounding $\tilde N$.   Using the inequality  ${n\choose k}\leq (ne/k)^k$, we have
\begin{align*}
{M' \choose \lceil\lambda M'\rceil}&\leq \left(\frac{eM'}{\lceil \lambda M'\rceil} \right)^{\lceil \lambda M'\rceil}\leq \left(\frac{eM'}{ \lambda M'} \right)^{\lceil \lambda M'\rceil}=\left(\frac{e}{\lambda}\right)^{\lceil \lambda M'\rceil}
\end{align*}
So we have
\begin{align*}
\tilde N&\geq M'^{-1} \left(\frac{\lambda}{e}\right)^{\lceil \lambda M'\rceil}\left( \frac{1-\epsilon}{2\epsilon}\right)^{\lceil\lambda M'\rceil}\\
&=M'^{-1}\left(2-\frac{\lambda}{2e}\right)^{\lceil \lambda M'\rceil}
\end{align*}
where the last equality holds by substituting the choice of $\epsilon$ in \eqref{eq:choice}. Using the fact $\lambda<1/2$, $\lceil \lambda M'\rceil\geq \lambda M'$ and $\epsilon M\leq  M' \leq \epsilon M+1$, we can further lower bound $\tilde N$ as 
\begin{align*}
\tilde N\geq \frac{1}{\epsilon M+1} \left(2-\frac{1}{4e}\right)^{\lambda \epsilon M}
\end{align*}
or equivalently
\begin{align*}
\log \tilde N &\geq -\log (\epsilon M +1)+\lambda\epsilon M a_0\\
&> -\log (M +1)+\lambda\epsilon M a_0\\
&\geq -\log (2M)+\lambda\epsilon M a_0
\end{align*}
where we define $a_0:=\log (2-1/(4e))$, and use the fact $\epsilon <1 $ and $M+1\leq 2M$ in the last two inequalities. Substituting the choices of $\lambda, \epsilon$ in \eqref{eq:choice} and $M=2^{n(C-\xi)}$, the above inequality becomes
\begin{align*}
\log \tilde N\geq -1 -n(C-\xi) +  \frac{\delta^2}{4eS^2}2^{n(C-\xi)}a_0
\end{align*}
Using the first inequality in \eqref{eq:S_upper_bound} to upper bound $S$, we have
\begin{align*}
\log \tilde N&\geq -1 -n(C-\xi) +  m\frac{\delta^2}{4e}2^{n(C-\xi)-n(C-\eta)}a_0\\
&= -1 -n(C-\xi) + m2^{n(\eta-\xi)}\frac{\delta^2a_0}{4 e}
\end{align*}
Now we can show that \eqref{eq:m} holds for large enough $n$ by showing the RHS in the above expression is larger than $m$ for large enough $n$. Indeed, we have
\begin{align*}
-1 -n(C-\xi) + m2^{n(\eta-\xi)}\frac{\delta^2a_0}{4 e}-m = m\left(2^{n(\eta-\xi)}\frac{\delta^2a_0}{4e}-1 \right) -1-n(C-\xi)\rightarrow \infty
\end{align*}
as $n\rightarrow \infty$ because $\eta>\xi$.

We have shown that with the choices in \eqref{eq:choice}, the error probability of the BFC code is given by $\lambda_1=\delta$ and $\lambda_2=2\delta$ which can be made arbitrarily small. By the choice $m=L_S(C-\eta, n)$ for all $\eta>0$, we have shown that $C$ is an achievable computation rate with the \rf $L_S$ defined in \eqref{eq:L_S}.
\end{proof}

\subsection{Proof of the converse}
\label{sec:proof_converse}

For the converse results, we follow the idea briefly discussed in \cite{ahlswede_general_2008}. The strategy is to show that any $(n,m,\mc F, \lambda_1, \lambda_2)$ BFC code can be converted to a standard identification code, and the known converse result on identification code will imply a converse result for the BFC code. Recall the definition of the $(n,N,\lambda_1,\lambda_2)$ ID code in Section \ref{sec:Ahlswede_models}. Define $N^*(n,\lambda_1,\lambda_2)$ to be the maximal $N$ such that an $(n,N, \lambda_1,\lambda_2)$ ID code exists. We have the following known converse result for ID codes.

\begin{theorem}[\cite{ahlswede_identification_1989,ahlswede_strong_2002}]
Let $\lambda_1,\lambda_2>0$ such that $\lambda_1+\lambda_2<1$. Then for every $\delta>0$ and every sufficiently large $n$
\begin{align*}
N^*(n, \lambda_1,\lambda_2)\leq 2^{2^{n(C+\delta)}}
\end{align*}
where $C$ is the Shannon capacity of the channel.
\label{thm:converse_id}
\end{theorem}

To prove the result in Theorem \ref{thm:general_converse}, we in fact show that there is a subset $\mc F\subset \wcf(S)$ such that any BFC code for the set $\mc F$ satisfies the statement in Theorem \ref{thm:general_converse}. Specifically, we prove the following result.

\begin{theorem}[Converse for a subset]
Assume the channel $W(\cdot | x^n)$ has a Shannon capacity $C>0$. Then there exists a set of constant weight functions $\mc F\subset \wcf(S)$, such that the computation capacity $\cbfc$ for $\mc F$ satisfies the statements in Theorem \ref{thm:general_converse}.
\label{thm:converse_subset}
\end{theorem}

Since  $\mc F$ is contained in $\wcf(S)$, any converse result on the code $(n,m,\mc F,\lambda_1,\lambda_2)$ automatically implies a converse result on the code  $(n,m,\wcf(S),\lambda_1,\lambda_2)$, which proves Theorem \ref{thm:general_converse}. The proof of Theorem \ref{thm:converse_subset} relies on the following Lemma \ref{lemma:bcf_to_id}, which is proved using Lemma \ref{lemma:gilbert_cw} (Gilbert's bound for constant weight sequences) in Appendix. Lemma \ref{lemma:bcf_to_id} states that we can find a subset $\mc F\subset\wcf(S)$ which has essentially the same size as $\wcf(S)$, but whose elements have small overlaps in their pre-images. The latter property allows us to convert a BFC code into a good identification code.

\begin{lemma}
For any integer $m\geq 1$, every integer $S\in[1, 2^{m-1}]$ and every $\alpha\in(0,1)$, there exists a set of constant weight functions $\mc F\subset \wcf(S)$ satisfying the following properties
\begin{itemize}
\item For any distinct $f_i, f_j\in \mc F$, it holds that $|f_i^{-1}[1]\cap f_j^{-1}[1]|\leq \alpha S$
\item If  $S\leq c2^{\gamma  m}$ for some $c>0$ and $\gamma \in[0,1)$,  $\log |\mc F|\geq \frac{\alpha(1-\gamma)}{4}Sm$ for $m$ large enough.
\item If $S=c2^{\gamma(m)m}$ with  $\gamma(m)\rightarrow 1$ for some $c\in(0,1/2]$, $\log |\mc F|\geq \frac{\alpha S}{2}(\log(1/c)+(1-\gamma(m))m)$ for $m$ large enough.
\end{itemize}
\label{lemma:bcf_to_id}
\end{lemma}
\begin{proof}
Any function $f\in\mc F_m$ can be represented by  a  binary sequence $\ve f$ of length $2^m$, where the $i$-th position of $\ve f$ is $1$ if and only if $f(i')=1$. Here $ i'$ is the dyadic representation of $i$, i.e., the unique binary sequence of length $m$ such that $i=\sum_{j=0}^{m-1}2^j i'_{j+1}$.  In other words, $\ve f$ is defined via $\ve f_i:=f(i')$. It is clear that if $f$ has a Hamming weight $S$, the corresponding vector $\ve f$ also has Hamming weight $S$. So we can apply Lemma \ref{lemma:gilbert_cw} by identifying $T=2^m$ and $G=\mc F$, which establishes
\begin{align*}
\log |\mc F|\geq \log {2^m\choose S}-\log {2^m\choose \lceil (1-\alpha)S\rceil-1 }-S
\end{align*}
for all $S\leq \frac{1}{2}2^m$. To further simplify the lower bound, we use the standard binomial bounds $(n/k)^k\leq {n\choose k}\leq (en/k)^k$ to obtain
\begin{align*}
\log |\mc F|&\geq S\log (2^m/S)-(\lceil (1-\alpha)S\rceil-1 )\log \left(\frac{e2^m}{\lceil (1-\alpha)S\rceil-1 } \right)-S\\
&\geq S(m-\log S)- (1-\alpha)S (m+\log e-\log ((1-\alpha)S) )-S
\end{align*}
where the last inequality holds because $x\log \frac{2^m}{x}$ is increasing for $x<2^m$, and   $\lceil (1-\alpha)S\rceil-1 \leq (1-\alpha)S$. Simplifying the expression we have
\begin{align*}
\log |\mc F|\geq S(\alpha(m-\log S) -K_\alpha)
\end{align*}
where $K_\alpha:=1+(1-\alpha)\log\frac{e}{1-\alpha}$. Notice for $m$ large enough, we have $K_\alpha\leq \frac{\alpha}{2}(m-\log S)$ as $m-\log S\rightarrow \infty$ as $m\rightarrow \infty$, this shows
\begin{align}
\log |\mc F|\geq \frac{\alpha}{2}S(m-\log S)
\label{eq:converse_2nd_last}
\end{align}
for large enough $m$.  

If $S\leq c2^{\gamma m}$ for some $\gamma\in [0,1)$, we have $m-\log S\geq (1-\gamma)m-\log c$. Since we have $\log c\leq (1-\gamma)m/2$ for large enough $m$, \eqref{eq:converse_2nd_last} implies that
\begin{align*}
\log |\mc F|\geq \frac{\alpha(1-\gamma)}{4}Sm
\end{align*}
for large enough $m$.

If $S=c 2^{\gamma(m)m}$ with $\gamma(m)\rightarrow 1$  for some $c\in(0,1/2]$, \eqref{eq:converse_2nd_last} implies that
\begin{align*}
\log |\mc F|\geq \frac{\alpha}{2} S (m-\log c-\gamma(m)m)=\frac{\alpha S}{2}(\log(1/c)+(1-\gamma(m))m)
\end{align*}
for large enough $m$.
\end{proof}

\begin{proof}[Proof of Theorem \ref{thm:converse_subset}]
We consider a set of functions $\mc F$ with properties specified  by Lemma \ref{lemma:bcf_to_id}, and consider  an  $(n,m, \mc F,\lambda_1,\lambda_2)$ BFC code characterized by $(Q_i)_{i\in\{0,1\}^m}, (D_j)_{j\in \intset{|\mc F|}}$  for this class of functions.

We first show that we can convert the above BFC code to an identification code as follows. Specifically, the new identification code is defined to be the encoder/decoder pairs $\{(\tilde Q_j, D_j), j\in\intset{|\mc F|}\}$ where encoder $\tilde Q_j$ is defined as
\begin{align*}
\tilde Q_j(x^n):=\frac{1}{S}\sum_{i\in f_j^{-1}[1]}Q_i(x^n)
\end{align*}
Now we characterize the error probability of this constructed identification code. The mis-identification error is 
\begin{align*}
\sum_{\ve x} \tilde Q_j(\ve x)W(D_j^c|\ve x)&=\sum_{\ve x}\frac{1}{S}\sum_{i\in f_j^{-1}[1]}Q_i(\ve x)W(D_j^c|\ve x)\\
&=\frac{1}{S}\sum_{i\in f_j^{-1}[1]}\sum_{\ve x}Q_i(\ve x)W(D_j^c|\ve x)\\
&\leq \frac{1}{S}\sum_{i\in f_j^{-1}[1]} \lambda_1\\
&=\lambda_1
\end{align*}
where the inequality holds because $f_j(i)=1$ hence $\sum_{\ve x}Q_i(\ve x)W(D_j^c|\ve x)\leq\lambda_1$, by the property of the BFC code. The last equality holds  because $|f_{j}^{-1}[1]|=S$ as $f_j$ has Hamming weight $S$.

The wrong-identification error is, for $\ell\neq j$, 
\begin{align*}
\sum_{\ve x} \tilde Q_\ell(\ve x)W(D_j|\ve x)&=\sum_{\ve x}\frac{1}{S}\sum_{i\in f_\ell^{-1}[1]}Q_i(\ve x)W(D_j|\ve x)\\
&=\frac{1}{S}\left(\sum_{i\in f_\ell^{-1}[1]\cap f_j^{-1}[1]}\sum_{\ve x}Q_i(\ve x)W(D_j|\ve x)+\sum_{i\in f_\ell^{-1}[1]\backslash f_j^{-1}[1]}\sum_{\ve x}Q_i(\ve x)W(D_j|\ve x)\right)\\
&\stackrel{(a)}{\leq} \frac{1}{S}\left(|f_\ell^{-1}[1]\cap f_j^{-1}[1]|+  \sum_{i\in f_\ell^{-1}[1]\backslash f_j^{-1}[1]}\lambda_2\right)\\
&\stackrel{(b)}{\leq} \frac{|f_\ell^{-1}[1]\cap f_j^{-1}[1]|}{S}+\lambda_2\\
&\stackrel{(c)}{\leq} \alpha+\lambda_2
\end{align*}
where in (a) we use the upper bound $\sum_{\ve x}Q_i(\ve x)W(D_j|\ve x)\leq 1$ and the BFC code error property, step (b) follows $|f_\ell^{-1}[1]\backslash f_j^{-1}[1]|\leq S$, and step (c) follows from the first property stated in Lemma \ref{lemma:bcf_to_id}. 

Therefore, we have shown that  an $(n,m, \mc F, \lambda_1,\lambda_2)$ BFC code (with properties of $\mc F$ prescribed in Lemma \ref{lemma:bcf_to_id}) implies the existence of an $(n, |\mc F|, \lambda_1, \alpha+\lambda_2)$ ID code for all $\alpha\in(0,1)$. Now fix some $\alpha\in(0,1-\lambda_1-\lambda_2)$.  Theorem \ref{thm:converse_id} states that for any $\lambda_1,\lambda_2$ such that $\lambda_1+\lambda_2<1$,  it holds that
\begin{align}
\log |\mc F|\leq 2^{n(C+\delta)}
\label{eq:converse_F}
\end{align}
for every $\delta>0$ and every sufficiently large $n$. 

Now we combine the inequality \eqref{eq:converse_F} with the result in Lemma \ref{lemma:bcf_to_id} under the condition  $S\leq c2^{\gamma m}$, which states
\begin{align*}
 \frac{\alpha(1-\gamma)}{4}Sm\leq 2^{n(C+\delta)}
\end{align*}
which is equivalent to 
\begin{align*}
Sm\leq \frac{4}{\alpha(1-\gamma)}2^{n(C+\delta)}\leq 2^{n(C+\delta')}
\end{align*}
for all $\delta'>\delta>0$ and  sufficiently large $n$. 

Now consider the case  $S=c2^{\gamma(m)m}$ with $\gamma(m)=1$ or $\gamma(m)\rightarrow 1$ where $c\in(0,1/2]$. Combining the inequality \eqref{eq:converse_F} with the result in Lemma \ref{lemma:bcf_to_id} under this case which states
\begin{align*}
\frac{\alpha S}{2}(\log(1/c)+(1-\gamma(m))m)\leq 2^{n(C+\delta)}
\end{align*}
or equivalently
\begin{align*}
\log\frac{\alpha}{2} +\log c + \gamma(m)m + \log(\log(1/c)+(1-\gamma(m))m) \leq n(C+\delta)
\end{align*}
As $\gamma(m)\rightarrow 1$, for any $\epsilon>0$, we can find $m$ large enough such that $\gamma(m)\geq (1-\epsilon)$. Then the above inequality implies that
\begin{align*}
\log\frac{\alpha}{2} +\log c +(1-\epsilon)m+\log (\log (1/c))\leq n(C+\delta)
\end{align*}
or equivalently
\begin{align*}
m\leq \frac{n(C+\delta)-\log(\log(1/c)c\alpha/2)}{(1-\epsilon)}
\end{align*}
It is straightforward to see that for any $\delta'>0$, we can choose appropriate $\epsilon$ and $\delta$  such that
\begin{align*}
m\leq n(C+\delta')
\end{align*}
for all sufficiently large $n$. This proves the claimed result. 
\end{proof}

\appendix

\subsection{Proof of Lemma \ref{lemma:equivalence}}\label{appendix:proof_equivalence}

To prove the achievability, we assume without loss of generality that $R$ is an achievable computation rate with rate function $L_1$. This means that for any $\eta>0$,   there exists an $(n,m,\mc F, \lambda_1,\lambda_2)$ BFC code where $m$ satisfies $\log m\geq \log L_1(R-\eta, n)$. Now we show that $R$ is also achievable with rate function $L_2$.  Since $L_2$ is a valid rate function, Definition \ref{def:valid_rate_function} implies that for every $R_1>R_2>0$, there exist constants
$\kappa>1$ and $n_0\in\mathbb{N}$ such that for all $n\geq n_0$, it holds that $   L_2(R_1,n)    \geq    \kappa L_2(R_2,n)$,
or equivalent 
\begin{align}
\log L_2(R_1,n)\geq \log\kappa +\log L_2(R_2, n)
\label{eq:multiplicative-separation}
\end{align}
Furthermore, since $L_1$ and $L_2$ are equivalent rate functions, Definition \ref{def:equivalent_rate_function} implies that for all $\epsilon>1$ and large enough $n$, we have $\frac{L_2(R,n)}{L_1(R,n)}\leq \epsilon$.  Now choosing $n$ large enough such that $\epsilon\leq \kappa$, we have
\begin{align*}
\log m\geq \log L_1(R-\eta, n)&=\log L_2(R-\eta,n) -\log \frac{L_2(R-\eta, n)}{L_1(R-\eta,n)}\\
&\stackrel{(a)}{\geq} \log L_2(R-2\eta, n)+\log \kappa  -\log \frac{L_2(R-\eta, n)}{L_1(R-\eta,n)}\\
&\geq \log L_2(R-2\eta, n)+\log \kappa  -\log \epsilon\\
&\geq \log L_2(R-2\eta, n)
\end{align*}
where step $(a)$ follows from \eqref{eq:multiplicative-separation}.

For the converse part, assume that for large enough $m$ and $n$,  there exists an $(n,m,\mc F, \lambda_1,\lambda_2)$ BFC code,  $m$ must satisfy $\log m\leq \log L_1(R+\delta, n)$. Since $L_1$ and $L_2$ are equivalent rate functions,  for all $\epsilon>1$ and large enough $n$, it holds that $\frac{L_1(R,n)}{L_2(R,n)}\leq \epsilon$.  Then for large enough $n$ we have
\begin{align*}
\log m&\leq \log L_2(R+\delta, n) + \log \frac{L_1(R+\delta,n)}{L_2(R+\delta,n)}\\
&\stackrel{(a)}{\leq} \log L_2(R+\delta/2,n)-\log \kappa + \log \epsilon\\
&\leq \log L_2(R+\delta/2,n)
\end{align*}
where step $(a)$ follows from \eqref{eq:multiplicative-separation} and in the last step we choose $n$ large enough such that $\epsilon\leq \kappa$.

\subsection{The proof of equivalent rate functions for Case 6 in Proposition \ref{prop:special_S}}
\label{appendix:equivlent_rate_functions}

Let $S(m):= c2^{m/\log m}$ for some $c>0$. Define $    L_1(R,n):=S^{-1}\left(2^{Rn}\right)$ and $L_2(R,n)=Rn\log n$. We prove that the two rate functions are equivalent. Define $f(x):=\frac{x}{\log x}$ so $S(m)=2^{f(m)}$.  We have $f'(x)=\frac{\log x-\log e}{(\log x)^2}$, so $f(x)$ is strictly increasing for  $x>e$, therefore has a strictly increasing inverse for large $x$ (i.e. $x>e$).

We first derive upper and lower bounds on $f^{-1}$. We have
\begin{align*}
f(x\log x)=\frac{x\log x}{\log x+\log\log x}\leq \frac{x\log x}{\log x}=x
\end{align*}
Since $f$  is strictly increasing for large $x$, we have  
\begin{align}
    f^{-1}(x)
    \geq
    x\log x.
    \label{eq:f-inverse-lower}
\end{align}

For the upper bound, define $ u(x):= x\bigl(\log x+2\log\log x\bigr)$. Then
\begin{align*}
    f(u(x))
    &=
    \frac{
        x\bigl(\log x+2\log\log x\bigr)
    }{
        \log x+
        \log\bigl(\log x+2\log\log x\bigr)
    }.
\end{align*}
For sufficiently large $x$, we have $(\log x)^2\geq \log x+2\log\log x$, so
\begin{align*}
f(u(x))\geq \frac{x(\log x+2\log\log x)}{\log x+2\log\log x}=x
\end{align*}
Since $f$ is strictly increasing, this gives
\begin{align}
    f^{-1}(x)
    \leq u(x)=
    x\bigl(\log x+2\log\log x\bigr).
    \label{eq:f-inverse-upper}
\end{align}

Due to the fact $S(m)=c2^{f(m)}$, we have $S^{-1}(y)=f^{-1}(\log y-\log c)$ for large enough $y$. Hence the rate function (for large enough $n$) is 
\begin{align*}
L_1(R,n):=S^{-1}(2^{nR})=f^{-1}(Rn-\log c)
\end{align*}
Applying \eqref{eq:f-inverse-lower} and
\eqref{eq:f-inverse-upper} gives
\begin{align}
    (Rn-\log c)\log (Rn-\log c)
    \leq
    L_1(R,n)
    \leq
    (Rn-\log c)\bigl(\log (Rn-\log c)+2\log\log (Rn-\log c)\bigr).
    \label{eq:L1-sandwich}
\end{align}

Dividing the lower bound in \eqref{eq:L1-sandwich} by
$L_2(R,n)=Rn\log n$, we have
\begin{align}
    \frac{Rn-\log c}{Rn}\cdot \frac{\log(Rn-\log c)}{\log n}= \left(1-\frac{\log c}{Rn}\right )\cdot \left(1- \frac{\log(R-\log c/n)}{\log n} \right)\rightarrow 1   \label{eq:lower-ratio-limit}
\end{align}

Similarly, dividing the upper bound in \eqref{eq:L1-sandwich} by
$Rn\log n$ yields
\begin{align}
\frac{   (Rn-\log c)\bigl(\log (Rn-\log c)+2\log\log (Rn-\log c)\bigr)}{Rn\log n}\rightarrow 1
 \label{eq:upper-ratio-limit}
\end{align}
It follows from \eqref{eq:L1-sandwich},
\eqref{eq:lower-ratio-limit}, and
\eqref{eq:upper-ratio-limit} that $    \lim_{n\rightarrow\infty}    \frac{L_1(R,n)}{L_2(R,n)}=1$. Therefore, $L_1$ and $L_2$ are equivalent rate functions.

We also need to check that both $L_1$ and $L_2$ are valid rate functions according to Definition \ref{def:valid_rate_function}. It is clear that $L_2$ is a valid rate function. To show $L_1$ is also valid, take any $R_1>R_2>0$, and invoke the upper and lower bounds in \eqref{eq:L1-sandwich}, we have
\begin{align*}
\frac{L_1(R_1, n)}{L_1(R_2,n)}\geq \frac{ (R_1n-\log c)\log (R_1n-\log c)}{(R_2n-\log c)\bigl(\log (R_2n-\log c)+2\log\log (R_2n-\log c)\bigr)}\rightarrow \frac{R_1}{R_2}>1.
\end{align*}

\subsection{Proof of Proposition \ref{prop:maximal_code}}

The proof follows that of \cite[Proposition 1]{ahlswede_identification_1989}. For any given $A_1$, we count the number of $A$ such that $|A_1\cap A|\geq \lambda M'$:
\begin{align*}
\sum_{i=\lceil \lambda M'\rceil}^{M'} \binom{|Z|-M'}{M'-i}\binom{M'}{i}
\end{align*}
For $\lambda<1/2$ and $1/\epsilon>6$, the first summand is the maximal one. Then we can upper the above expression as
\begin{align*}
M'\binom{|Z|-M'}{M'-\lceil \lambda M'\rceil}\binom{M'}{\lceil \lambda M'\rceil}\leq M' \binom{|Z|}{M'-\lceil \lambda M'\rceil} {M' \choose \lceil\lambda M'\rceil}=:T.
\end{align*}

As there are $\binom{|Z|}{M'}$ sets of cardinality of $M'$, so if $T<\binom{|Z|}{M'}$, then there must exist at least one $A_2$ such at $|A_1\cap A_2|<\lambda M'$.  Similarly, if $2T< \binom{|Z|}{M'}$, we can find another set $A_3$ with $|A_3|=M'$, $|A_3\cap A_2|<\lambda M'$ and $|A_3\cap A_1|<\lambda M'$. In general, if  it holds that $(N-1)T<\binom{|Z|}{M'}$, we can find $N$ sets with the desired property.  Hence it is easy to see that a family of sets $A_1,\ldots, A_N$ exists with the choice
\begin{align*}
N:=\left\lceil\binom{|Z|}{M'}T^{-1}\right\rceil
\end{align*}
Clearly we have
\begin{align*}
N\geq \binom{|Z|}{M'}T^{-1}
\end{align*}
As shown in \cite[Proposition 1]{ahlswede_identification_1989}, we have
\begin{align*}
\binom{|Z|}{M'}T^{-1}&=H(\lambda, M') M'^{-1}\prod_{i=1}^{\lceil \lambda M'\rceil}\frac{|Z|-M'+i}{M'-\lceil\lambda M'\rceil +i}
\end{align*}
Recall that $M':=\lceil \epsilon |Z|\rceil$. Notice that for $i= 1,\ldots, \lceil \lambda M'\rceil$, we have
\begin{align*}
\frac{|Z|-M'+i}{M'-\lceil\lambda M'\rceil +i}&\geq \frac{|Z|-M'+1}{M'-\lceil\lambda M'\rceil +\lceil \lambda M'\rceil}\\
&\geq \frac{|Z|-M'+1}{M'}\\
&\geq \frac{|Z| -\epsilon |Z|}{M'}\\
&\geq \frac{|Z|-\epsilon|Z|}{2\epsilon|Z|}\\
&=\frac{1-\epsilon}{2\epsilon}
\end{align*}
where the last inequality holds because $M'\leq \epsilon|Z|+1\leq 2\epsilon|Z|$, as it is assumed that $\epsilon |Z|\geq 1$. This implies
\begin{align*}
N\geq H(\lambda, M')M'^{-1}\left( \frac{1-\epsilon}{2\epsilon}\right)^{\lceil \lambda M'\rceil}
\end{align*}

%


\subsection{Proof of Lemma \ref{lemma:error_probability}}
\label{appendix:error_probability}

We first analyze the false negative error which can happen when $f_j(i)=1$, or equivalently 
\begin{align}
i\in f_j^{-1}[1]
\label{eq:error_analysis_1}
\end{align}
In this case, we have
\begin{align*}
Q_iW(\mc D_j^c)&=\sum_{\ve x\in\mc X^n}Q_i(\ve x)W(\mc D_j^c|\ve x)\\
&\stackrel{\eqref{eq:encoder}}{=}\sum_{\ve x\in\mc X^n}\ve 1_{\ve x\in A_i}\frac{1}{|A_i|} W(\mc D_j^c|\ve x)\\
&=\frac{1}{|A_i|}\sum_{\ve u_k\in A_i}W(\mc D_j^c|\ve u_k)\\
&\stackrel{\eqref{eq:Dj_def}, \eqref{eq:error_analysis_1}}{\leq} \frac{1}{|A_i|}\sum_{\ve u_k\in A_i}W(\mc E_k^c|\ve u_k)\\
&\stackrel{\eqref{eq:good_channel_code}}{\leq}  \frac{1}{|A_i|}\sum_{\ve u_k\in A_i}\delta= \delta
\end{align*}
In other words, the false negative error is  always upper bounded by $\delta$, irrespective of $S$.

Now consider the false positive error which can happen when $f_j(i)=0$.  In this case we have
\begin{align*}
Q_iW(\mc D_j)&=\sum_{\ve x\in\mc X^n}Q_i(\ve x)W(\mc D_j|\ve x)\\
&\stackrel{\eqref{eq:encoder}}{=}\sum_{\ve x\in\mc X^n}\ve 1_{\ve x\in A_i}\frac{1}{|A_i|} W(\mc D_j|\ve x)\\
&=\frac{1}{|A_i|}\sum_{\ve u_k\in A_i}W(\mc D_j|\ve u_k)\\
&=\frac{1}{|A_i|}\left( \sum_{\ve u_k\in A_i\bigcap B_j}W(\mc D_j|\ve u_k) +\sum_{\ve u_k\in A_i\backslash B_j}W(\mc D_j|\ve u_k)\right)\\
&\leq \frac{1}{|A_i|}\left( |A_i\cap B_j| +\sum_{\ve u_k\in A_i\backslash B_j}W(\mc D_j|\ve u_k)\right)
\end{align*}
where the last inequality holds as $W(\mc D_j|\ve u_k)\leq 1$. For $\ve u_k\notin B_j$, it holds that $\mc E_k\cap \mc D_j=\emptyset$ by the construction of $\mc D_j$ in \eqref{eq:Dj_def}. This means $\mc D_j\subseteq \mc E_k^c$ hence  for $\ve u_k\in A_i\backslash B_j$ we have
\begin{align*}
W(\mc D_j|\ve u_k)\leq W(\mc E_k^c|\ve u_k)\leq \delta.
\end{align*}
This implies
\begin{align}
Q_iW(\mc D_j)&\leq \frac{|A_i\cap B_j|}{|A_i|} + \frac{\delta |A_i\backslash B_j|}{|A_i|}\nonumber	\\
&\stackrel{(a)}{\leq} \frac{\sum_{\ell \in f_j^{-1}[1]} |A_i\cap A_\ell|}{|A_i|}+\delta\nonumber\\
&\stackrel{(b)}{<}   \frac{\sum_{\ell \in f_j^{-1}[1]} \lambda M'}{M'}+\delta\nonumber\\
&\stackrel{(c)}{\leq} S \lambda +\delta\label{eq:lambda_2}
\end{align}
where $(a)$ follows because of the definition of $B_j$ after \eqref{eq:Dj_def} and an application of the union bound,  $(b)$ follows due to the property of the sets $\{A_i\}_i$ by Proposition \ref{prop:maximal_code}, that is, $|A_i \cap A_j|<\lambda M'$ for all $ i\neq j$. and $(c)$ follows that the Hamming weight of the function $f_j$ is smaller or equal to $S$, namely $|f_j^{-1}[1]|\leq S$.

\subsection{Lemma \ref{lemma:gilbert_cw} and its proof}

The following lemma is used to prove Lemma \ref{lemma:bcf_to_id}, which will be used in the proof of the converse result in Theorem \ref{thm:converse_subset}.

\begin{lemma}[Gilbert's bound for constant weight sequences]\label{lemma:gilbert_cw}
Let  $T, S$ be two natural numbers such that $1\leq S\leq T/2$. For any $\alpha\in(0,1)$,  there exists a set $G$ of $T$-length binary sequences with the following properties
\begin{itemize}
\item  for any distinct $a,b\in G$, $|a\cap b|\leq \alpha S$
\item $\log |G|\geq \log {T\choose S}-\log {T\choose \lceil (1-\alpha)S\rceil-1 }-S$
\end{itemize}
\end{lemma}

\begin{proof}
By the Gilbert bound for constant weight sequences with a weight $S$, we know from  \cite[Theorem 7]{graham_lower_1980} that for any integer $0\leq d\leq S$, there exists a set $G$ of binary sequences where $d_H(a,b)\geq 2d$ for any $a,b\in G, a\neq b$ and
\begin{align*}
|G|\geq \frac{{T\choose S}}{\sum_{i=0}^{d-1}{S \choose i}{T-S \choose i}}
\end{align*}

To see the first claim, notice that two sequences $a,b$ with weight $S$ and Hamming distance $2\delta$ has an intersection $|a\cap b| = S-\delta$. To see this, we count the total weights of the two sequences in two different ways. First note that the total weights of the two sequences is $2S$, as both sequences have weight $S$. On the other hand, the total weights can also be expressed as $2|a\cap b|+2\delta$, where $2|a\cap b|$ is the number of positions where the two sequences take the value $1$, and $2\delta$ is the number of positions where only one of the two sequences takes the value $1$.  We let $d=\lceil (1-\alpha)S\rceil$.  Then it holds that $|a\cap b|=S-d_H(a,b)/2\leq  S-d\leq \alpha S$ as $d\geq (1-\alpha)S$.

To show the lower bound, notice that we have for $i=1,\ldots, d-1$
\begin{align*}
\sum_{i=0}^{d-1}{S \choose i}{T-S \choose i}&\leq \sum_{i=0}^{d-1}{S \choose i}{T \choose i}\\
&\stackrel{(a)}{\leq}  \sum_{i=0}^{d-1}{S \choose i}{T \choose d-1}\\
&\leq 2^S {T\choose d-1}
\end{align*}
where (a) holds because for $i=1,\ldots, d-1$ with $d-1\leq (1-\alpha)S\leq S\leq T/2$, we have ${T\choose i}\leq {T\choose d-1}$. 

\end{proof}

\bibliographystyle{IEEEtran}
\bibliography{BFC_journal}

\end{document}